%% file: main_body.tex
\documentclass[11pt]{article}

\input{preamble_usepackages.tex}

\input{preamble_newcommands.tex}
\input{preamble_newtheorem.tex}

\begin{document}
\title{$\system$: A Secure Remote-Control System for IoT Devices from\\Centralized Multi-Designated Verifier Signatures\thanks{This is the full version of a paper that will appear in ISPEC 2023~\cite{ispecWYS23}.}}

\author[1,2]{Yohei Watanabe}
\author[2,3]{Naoto Yanai}
\author[4]{Junji Shikata}

\affil[1]{
	The University of Electro-Communications, Tokyo, Japan.
	}
\affil[2]{
	Japan Datacom Co., Ltd., Tokyo, Japan.
	}
\affil[3]{
	Osaka University, Osaka, Japan.
	}
\affil[4]{
	Yokohama National University, Yokohama, Japan.
	\vspace*{.5em}
	}
\affil[ ]{
	\texttt{watanabe@uec.ac.jp}, \ \ 
	\texttt{yanai@ist.osaka-u.ac.jp}, \ \
	\texttt{shikata@ynu.ac.jp}
	}

\maketitle

\begin{abstract}
IoT technology has been developing rapidly, while at the same time, notorious IoT malware such as Mirai is a severe and inherent threat. 
We believe it is essential to consider systems that enable us to remotely control infected devices in order to prevent or limit malicious behaviors of infected devices.
In this paper, we design a promising candidate for such remote-control systems, called \emph{$\system$ (REmote-Control System for IoT devices)}.
$\system$ allows a systems manager to designate an arbitrary subset of all IoT devices in the system, and every device can confirm whether or not the device itself was designated;
if so, the device executes a command given by the systems manager. 
Towards realizing $\system$, we introduce a novel cryptographic primitive called \emph{centralized multi-designated verifier signatures} (CMDVS).
Although CMDVS works under a restricted condition compared to conventional MDVS, it is sufficient for realizing $\system$.
We provide an efficient CMDVS construction from any approximate membership query structures and digital signatures, yielding compact communication sizes and efficient verification procedures for $\system$.
We then discuss the feasibility of $\system$ through the cryptographic implementation of the CMDVS construction on a Raspberry Pi.
Our promising results demonstrate that the CMDVS construction can compress communication size to about 30\% compared to a trivial construction, and thus its resulting $\system$ becomes three times faster than a trivial construction over typical low-power wide area networks with an IoT device. 
\end{abstract}

\thispagestyle{empty}

\section{Introduction} 
\label{sec:Introduction}

Internet-of-Things technologies have been spreading rapidly and enriching our lives.
According to a Cisco report~\cite{Cisco2014}, tens of billions of IoT devices are expected to be deployed over the next few years.
On the other hand, along with the rapid development of IoT technologies, we have to focus our efforts on cybersecurity, though there are several constraints on that in the context of IoT devices.
For example, most IoT devices, unfortunately, do little to protect the data stored inside, mostly likely due to the development cost and restricted resources.
This has a profound effect on the real world; 
for instance, a notorious IoT malware `Mirai' infected many IoT devices, turning them into botnets.
The botnets infected nearly 65,000 IoT devices in its first 20 hours~\cite{antonakakis2017understanding}.
The widespread outbreak of Mirai had a considerable impact on the world.
As described above, most IoT devices do not have sufficient resources to implement and deploy security functions for each specific security threat~\cite{noor2019current}.
Hence, there seem to be no versatile solutions~\cite{bertino2017botnets}.

One possible approach is to design cryptographic schemes that can be used in cooperation with existing methods such as controlling~\cite{ainaWYS21,xu2019dominance,suzaki2020reboot} or surveillance~\cite{nguyen2019diot,raza2013svelte,lei2020secwir} of individual devices.
Cryptographic schemes can provide \emph{provable security} that theoretically guarantees the security of a cryptographic protocol through mathematical proofs.

In this paper, we present a novel system based on cryptography, \emph{$\system$} (REmote-Control System for IoT devices), which has an arbitrary subset of all IoT devices and executes any commands remotely and securely. 
The most likely scenario is to disable compromised IoT devices, e.g., those infected with malware.
$\system$ allows such devices to be brought to a halt as soon as possible. 
It is expected to, for example, stop and reboot malware-infected devices all at once, whereby a sender can communicate with many devices simultaneously with a single piece of data.

We note that the efficient design of $\system$ is \textit{non-trivial}. 
One might think $\system$ can be realized with a standard digital signature, regarding an arbitrary subset of devices' identifiers as a single message and signing it. 
However, it is insufficient because the communication size is linear in the size of the subset. 
Since IoT devices are resource-constrained~\cite{iftikhar2021resourceconstrained}, their battery life is also limited. 
Even if the latency on a CPU is small enough, the communication should be used sparingly to avoid consuming energy too quickly as well~\cite{usenixKHAP+19}.
Namely, we need to achieve the small communication size as well as the functionality to choose an arbitrary subset of receivers. 
As an advanced cryptographic approach, broadcast authentication~\cite{spPCTS00} might be employed;
it can broadcast a single piece of data to many receivers, i.e., IoT devices, with data authenticity for controlling them. 
However, existing broadcast authentication schemes~\cite{spCP10,ccsPer01,spPCTS00,tifsShi17,tcsSW01,baee2022ali} except for a recent work~\cite{ainaWYS21} cannot support the functionality that a sender chooses an arbitrary subset of receivers. 
Though the only exception~\cite{imaccKWS21}, i.e., the broadcast authentication scheme that supports such functionality, may be applied to $\system$, it still has the major drawback of communication sizes since it just combines individual authenticators for all designated devices.


To this end, we propose a novel cryptographic scheme, centralized multi-designated verifier signatures (CMDVS), as a core primitive for $\system$. 
CMDVS is an extension of multi-designated verifier digital signature schemes~\cite{ipsLV07,nssZAYS12,icicsLV04,tccDHM+20}.
Unlike conventional schemes, anyone can be both a signer and a verifier, while entities have completely different roles in CMDVS; there is only one sender and many verifiers.
Although CMDVS works under more restricted conditions than conventional MDVS, it can be constructed efficiently and is sufficient for realizing $\system$. 
We define the security of CMDVS formally and then propose an efficient CMDVS construction from any approximate membership query (AMQ) structure and digital signatures, which yields an efficient design for $\system$. 
The proposed construction is provably secure. 
Note that we show CMDVS provides more efficient communication sizes than the two trivial approaches described in the previous paragraph.

We also discuss the feasibility of $\system$ for IoT devices through the implementation of the proposed CMDVS construction with EdDSA~\cite{eddsa} and vacuum filters~\cite{vldbWZSQ19}, which is one of the efficient AMQ structures. 
We then demonstrate that the proposed CMDVS construction can compress communication size to about 30\% compared to the trivial approach with standard digital signatures. (Hereafter, we call this approach \textit{trivial construction}.) 
We also show that our scheme can also compress communication size to about 4\% compared to the broadcast-authentication-based approach~\cite{ainaWYS21}, which is simply called \textit{broadcast authentication} hereafter. 
Our promising results also show that, by virtue of the compression of the communication size, $\system$ is three times faster than the trivial construction and 25 times faster than the broadcast authentication over typical low-power wide area networks with a Raspberry Pi3 as an IoT device. 
We also evaluate the communication overheads and power consumption for low-power wide area networks. 
We have released our source code for reproducibility and subsequent work (\url{https://github.com/naotoyanai/fiilter-signature_ABA}).

To sum up, our primary goal is to design $\system$, and we make the following technical contributions: 
\begin{itemize}
    
    \item We propose CMDVS as a novel cryptographic primitive to instantiate $\system$. 
    We formally define and prove the security of the proposed construction. 

    \item We give an efficient instantiation of a CMDVS scheme from the (fine-tuned) Bloom filter. We provide theoretical performance analysis, and show our CMDVS instantiation is three times more compact than the trivial construction.
    
    
    \item Through an implementation, we experimentally demonstrate that the proposed CMDVS construction can compress communication size to about 30\% compared to the trivial construction and 4\% compared to the broadcast authentication. 
    We have released our code via GitHub.  
    
    \item We discuss the feasibility of $\system$, including the communication overheads for low-power wide area networks and the power consumption. 
\end{itemize}

\section{$\system$: \textsf{RE}mote-\textsf{C}ontrol \textsf{S}ystem for \textsf{IoT} Devices} \label{sec:System}
\subsection{System Setting} \label{sec:SystemGoal}
Suppose a large, simple system called $\system$ (REmote-Control System for IoT devices) among a systems manager and many IoT devices such as sensors and surveillance cameras below.

\mpara{$\system$: An Overview}
There are a systems manager and a number of IoT devices.
For some reason (e.g., based on data from outside sources such as device owner's request and information on vulnerable devices), the systems manager generates and broadcasts authenticated information in order to make only designated IoT devices execute a command $\cmd$ remotely and securely, while the devices themselves can detect a forgery of the authenticated information that aims to change the designated-device set and/or the command.

\mpara{Expected Applications}
We believe there are various applications of $\system$.
For example, it enables one to put only designated devices to sleep, e.g., in order to extend their operational lives.
At the same time, it prevents an adversary from forging the authenticated information on the `sleep' command and which devices are designated.
Besides, let us explain another important application:
the IoT devices usually communicate with each other via the Internet and could be infected with malware.
As explained in the introduction, it seems difficult to completely eliminate the chance of devices being infected with malware, and IoT malware spreads rapidly between IoT devices once the initial infection occurs.
Therefore, $\system$ can bring infected devices to a halt as soon as possible in order to prevent or limit malicious behavior by said devices (e.g., DDoS attacks), rather than preventing the initial infection.

\subsection{System Model} \label{sec:SystemModel}
Based on the above discussion, we formally define $\system$ as a protocol among the following entities:
a device owner $\Owner$, a systems manager $\SysM$, and IoT devices $\IoTvar$.
%
Let $\ID$ be a set of possible identifiers in the system, and $\IDact$ be an identifier set of activated devices, i.e., IoT devices taking part in the system.
We denote an identifier set of devices designated by $\SysM$ so that they execute a command $\cmd$ by $\IDdsg$.
We have $\IDdsg \subset \IDact \subset \ID$.

\mpara{System Overview}
Suppose that the device owner $\Owner$ manages many IoT devices $\set{\IoT{\id}}_{\id\in\IDact}$.
Note that $\Owner$ can dynamically add and remove IoT devices.
Let us explain the protocol overview as follows.
\begin{enumerate}
    \renewcommand{\labelenumi}{\ctext{\arabic{enumi}}}
    \item[\ctext{1}]
    %
    $\Owner$ sends $\SysM$ a request to have an arbitrary subset (i.e., $\IDdsg$) of all devices execute a command $\cmd$.
    \item[\ctext{2}] 
    %
    $\SysM$ generates an authenticated command $\acmd$, which is an authenticated version of $\cmd$ and contains the information on the designated devices $\IDdsg$, and broadcasts it to \emph{all} devices. 
    \item[\ctext{3}] 
    %
    All IoT devices $\set{\IoT{\id}}_{\id\in\IDact}$ (including non-designated ones) receive $\acmd$ and check its validity. If $\acmd$ is \emph{not} valid, the devices reject it and terminate the process.
    \item[\ctext{4}] 
    %
    If an IoT device $\IoT{\id}$ confirms that the authenticated command $\acmd$ is valid and directed at the device, $\IoT{\id}$ executes $\cmd$.
    Otherwise, i.e., if $\acmd$ is valid but does not designate $\IoT{\id}$, the device does nothing and terminates the process.
    %
\end{enumerate}

\subsection{Assumptions and Requirements} \label{sec:Assumptions_Requirements}
\para{Adversarial Model and Assumptions}
Suppose that the systems manager $\SysM$ broadcasts an authenticated command $\acmd$ to all devices $\set{\IoT{\id}}_{\id \in \IDact}$.
We assume an adversary $\A$ can eavesdrop, insert, delay, and modify all the transmitted information.
We also assume that $\A$'s main purpose is to maliciously modify authenticated commands so that some designated devices do not execute $\cmd$ and/or some non-designated devices execute $\cmd$.
More formally, we assume that $\A$ mainly aims to modify $\acmd$ in order to change a pair of $(\cmd,\IDdsg)$ to a different pair $(\cmd',\IDdsg')$ in order to accomplish any of the goals below:
\begin{enumerate}
    \renewcommand{\labelenumi}{(\alph{enumi})}
    \item At least one designated device $\IoT{\id}$ for $\id \in \IDdsg$ does not execute $\cmd$ as a regular process. 
    \item At least one designated device $\IoT{\id}$ for $\id \in \IDdsg$ executes $\cmd' \ (\neq \cmd)$ as a regular process.
    \item At least one non-designated device $\IoT{\id}$ for $\id \in \IDdsg'\setminus \IDdsg$ executes $\cmd'$, which might be the same as $\cmd$, as a regular process.
\end{enumerate}
Note that the above goals include that $\A$ tries to impersonate the systems manager $\SysM$ and create new (forged) authenticated commands.
However, we assume $\A$ is not capable of forging any CMDVS signature, which is a core element of authenticated commands $\acmd$, according to Def.~\ref{def:CMVDS_UF}, which will be defined later.

For simplicity, we assume that all devices receive the same information;
if authenticated commands are modified, all devices receive the modified ones.
We also note that preventing attacks in the physical layer is out of the scope, i.e., jamming. 
It can be prevented by existing techniques such as the spread spectrum~\cite{mpitziopoulos2007defending}.

\mpara{Requirements}
Following the discussion in the introduction and our system goal, the secure system for remotely controlling IoT devices, $\system$, should possess the following four properties.
\begin{itemize}
    \item \emph{Completeness}:
    Only designated devices $\set{\IoT{\id}}_{\id \in \IDdsg}$ execute a command $\cmd$ unless the corresponding authenticated command $\acmd$ is externally modified.
    In other words, any non-designated device $\IoT{\id}$, where $\id \in \IDact \setminus \IDdsg$, never executes $\cmd$ as long as it receives $\acmd$ as it is.
    The system might have allowable errors; a very small percentage of devices might not work as expected.
    This error seems likely in most large-scale applications.
    \item \emph{Integrity}:
    If an authenticated command $\acmd$ is externally modified, any device can detect it and reject $\acmd$.
    %
    %
    \item \emph{Scalablity}:
    The system allows a large number of IoT devices, e.g., up to a million.
    In particular, the size of authenticated commands should be small, i.e., it does not depend on the number of designated devices linearly.
    Ideally, it should be independent of the number of designated devices in the system.
    \item \emph{Light weight}:
    The devices' resources might be poor.
    Thus, the verification process executed by the devices should be efficient enough that, ideally, even microcomputers such as an ARM Cortex-M3 can run the process.
\end{itemize}

The first two requirements---completeness and integrity---are the fundamental properties to have $\system$ work well in practice.
The last two requirements---scalability and light weight---are also important properties for $\system$ since we focus on various IoT devices. including microcomputers.
Indeed, a trivial system can be constructed by an arbitrary digital signature or MAC: $\SysM$ just sends each designated IoT device a command $\cmd$ with its signature/MAC.
This trivial construction requires the $\bigO{d \cdot \secpar}$ communication size, where $d$ is the number of designated devices and $\secpar$ is a security parameter, whereas its verification process is lightweight since it requires only a single signature/MAC verification.
Hence, achieving both scalability and lightweight is another important goal for $\system$.

\section{Centralized Multi-Designated Verifier Signatures} \label{sec:CMDVS}

We introduce \emph{centralized MDVS} (CMDVS), which is a core cryptographic primitive for $\system$.
Unlike existing MDVS schemes~\cite{icicsLV04,tccDHM+20}, in CMDVS, we consider a situation where there are only one signer and multiple verifiers.
Note that CMDVS is not a special case of MDVS;
there are multiple users who are potential signers and/or verifiers in MDVS.

\mpara{Notations}
For any natural numbers $a,b\in\NN$ s.t.\ $a \le b$, $\set{a,\ldots,b}$ is denoted by $[a,b]$.
In particular, if $a=1$, we denote $[b] \ceq \{1,\ldots,b\}$.
For any real numbers $a,b \in \RR$ s.t.\ $a \le b$, let $(a,b]$ be a half-open interval.
Concatenation is denoted by $\|$.
For a finite set $\mathcal{X}$, we denote by $|\mathcal{X}|$ the cardinality of $\mathcal{X}$.
For any algorithm $\Ad$, $\mathsf{out} \gets \Ad(\mathsf{in})$ means that $\Ad$ takes $\mathsf{in}$ as input and outputs $\mathsf{out}$.
Throughout the paper, we denote by $\secpar$ a security parameter and consider probabilistic polynomial-time algorithms (PPTAs).
We say a function $\negl{\cdot}$ is negligible if for any polynomial $\poly{\cdot}$, there exists some constant $\secpar_0\in\mathbb{N}$ such that $\negl{\secpar}<1/\poly{\secpar}$ for all $\secpar\ge \secpar_0$.
In security games, a flag $\flag$, which indicates an adversary's winning condition, is initialized as zero.

\subsection{Syntax} \label{sec:CMDVS_Syntax}
First of all, a signer runs $\Setup$ to get a public parameter $\pp$ and a signing key $\sk$.
The signer can run $\KeyGen$ with $(\pp,\sk)$ to generate a verification key $\vrk{\id}$ for any $\id \in \ID$.
Let $\VerS$ be a verifier set, i.e., a set of identities whose key pairs have been generated by $\KeyGen$.
To create a signature $\sigma$ so that only a designated-verifier set $\Dv \subset \VerS$ accepts it, the signer executes $\Sign$ with $\sk$, $\Dv$, and a message $\msg$.
Each verifier can check the validity of $(\msg, \sigma)$ by $\Vrfy$ with $\pp$ and $\vrk{\id}$ if the verifier was designated by the signer, i.e., $\id \in \Dv$.
In other words, for any non-designated verifier $\id \notin \Dv$, $\Vrfy$ with $(\pp,\vrk{\id})$ outputs $\bot$ even if the pair $(\msg, \sigma)$ is a valid one.

CMDVS $\CMDVSall$ for an identity set $\ID$ is defined as follows.

\begin{enumerate}
    \renewcommand{\labelenumi}{--}
    %
    \item $\Setup(1^{\secpar}) \to (\pp,\sk)$: a probabilistic algorithm for setup.
    It takes a security parameter $1^{\secpar}$ as input, and outputs a public parameter $\pp$ and a signing key $\sk$.
    It initializes a verifier set $\VerS$.
    \item $\KeyGen(\pp, \sk,\id) \to \vrk{\id}$: an algorithm for verification-key generation.
    It takes $\pp$, $\sk$, an identity $\id \in \ID$ as input, and outputs a verification key $\vrk{\id}$ for $\id$.
    It also updates $\VerS \ceq \VerS \cup \set{\id}$.
    \item $\Sign(\sk, \Dv, \msg, \len) \to \sigma \ / \ \bot$: a signing algorithm.
    It takes $\sk$, a designated-verifier set $\Dv \subset \VerS$, a message $\msg \in \M$, and the maximum length of a signature $\len$ as input, and outputs the signature $\sigma$ for $\Dv$ or $\bot$, which indicates ``failure of signature generation.''
    \item $\Vrfy(\pp, \vrk{\id}, \msg, \sigma) \to \top \ / \ \bot$: a deterministic algorithm for verification.
    It takes $\pp$, $\vrk{\id}$, $\msg$ and $\sigma$ as input, and outputs $\top$ indicating ``accept'' or $\bot$ indicating ``reject.''
\end{enumerate}

\begin{remark}[On the Maximum Length $\len$ of Signatures] \label{rem:CMDVS_Syntax} {\rm
    CMDVS allows a signer to specify the maximum length $\len$ when generating the corresponding signature since we aim to design $\system$ so that it is compatible with various environments, including wireless ones, which often restricts bandwidth.
    The length specification feature enables us to generate signatures so that they fit in the channel's bandwidth.
    Indeed, although a trivial construction in Section~\ref{sec:Construction_Trivial} produces signatures whose length depends on the number of designated verifiers, the proposed generic construction in Section~\ref{sec:Construction_Generic} allows flexible parameter settings, i.e., a signer first fixes $\len$ and then chooses other parameters (see also Remark~\ref{rem:CMDVS_ParameterSetting}).
}\end{remark}

\subsection{Correctness and Security} \label{sec:CMDVS_Security}
We introduce the correctness property and security notions for CMDVS.

\mpara{Oracles}
We consider the following oracles.
Let $\VKList$ and $\qM$ be an array and a set, respectively, and they are initialized as empty ones.
\begin{itemize}
    \setlength{\itemsep}{0em}
    \setlength{\parskip}{0em}
    \item Key-generation oracle $\KGOracle{\cdot}$:
    For any $\id \in \ID$, it runs $\KeyGen(\pp, \sk,\id)$ to get $\vrk{\id}$.
    It adds $\id$ and $\vrk{\id}$ to $\VerS$ and $\VKList[\id]$, respectively, and returns $\vrk{\id}$.
    \item Signing oracle $\SOracle{\cdot}$:
    For any $(\Dv, \AB \msg, \AB \len)\in 2^{\VerS} \times \M \times \NN$, it returns $\Sign(\sk, \AB \Dv, \AB \msg, \AB \len)$.
    It adds $(\Dv, \AB \msg)$ to $\qM$ if $\Sign(\sk, \AB \Dv, \AB \msg, \AB \len)\neq \bot$. 
    %
\end{itemize}

\begin{remark}[On Provable Anonymity] \label{rem:CMDVS_Oracle} {\rm
    An adversary obtains all verification keys via the above key-generation oracle, i.e., all verification keys are public.
    Namely, unlike ordinary MDVS, we consider security against unbounded collusion of verifiers.
    In this setting, (provable) anonymity of designated verifiers~\cite{tccDHM+20}, which is an additional security notion for MDVS, cannot be achieved in principle since the verification algorithm works with only public information.
    It might be possible by restricting the range of verification keys that the adversary can get, though it would be expected to make CMDVS less efficient.
}\end{remark}

\mpara{Correctness}
The correctness property guarantees that each verifier correctly obtains the output of $\Vrfy$ algorithm unless signatures are maliciously modified.

\begin{definition}[Correctness] \label{def:CMDVS_Correctness} {\rm
Let $\CMDVS$ be a CMDVS scheme.
$\CMDVS$ is said to meet correctness if for any $\msg\in\M$, for any $\VerS \subset \ID$ such that $|\VerS| = \poly{\secpar}$, any $\Dv \subset \ID$, and for any $\id \in \VerS$, it holds that
\begin{align*}
    \left\{
    \begin{array}{ll}
        \Pr\LRbra{\Vrfy(\pp, \vrk{\id}, \msg, \sigma) \to \true} \ge 1-\negl{\secpar} & \text{if } \id \in \Dv,  \\
        \Pr\LRbra{\Vrfy(\pp, \vrk{\id}, \msg, \sigma) \to \false} \ge 1-\negl{\secpar} & \text{if } \id \in \VerS\setminus\Dv,
    \end{array}
    \right.
\end{align*}
where $(\pp,\sk) \gets \Setup(1^\secpar)$, $\vrk{\id}\gets\KeyGen(\pp,\sk,\id)$ for all $\id \in \VerS$, and $\sigma \ (\neq \bot) \gets \Sign(\sk,\Dv,\msg,\len)$.
}\end{definition}

\begin{remark}[On Designated-Verifier Sets] \label{rem:CMDVS_Correctness} {\rm
    As can be seen in Def.~\ref{def:CMDVS_Correctness}, a designated-verifier set $\Dv$ need not necessarily be a subset of the verification set $\VerS$.
    This means that the signer can designate identities before the corresponding verification keys are generated.
    Indeed, security definitions below do not restrict the range of a designated-verifier set $\Dv^\star$ chosen by an adversary.
}\end{remark}

\begin{figure}[t]
\begin{minipage}[t]{0.48\hsize}
    \centering
    \begin{algorithm}[H]
        \floatname{algorithm}{Experiment:}
        \caption{$\Exp{\CMDVS,\A}{\UF}(\secpar)$}
        \label{exp:CMDVS_UF}
        \begin{algorithmic}[1]
            \STATE $(\pp,\sk)\gets\Setup(1^\secpar)$
            \STATE $(\Dv^\star,\msg^\star,\sigma^\star)\gets \A^{\KGOra,\SOra}(1^\secpar,\pp)$ 
            \IF{$(\Dv^\star,\msg^\star)\notin\qM$} 
                \IF{$\exists \id^\star\in\Dv^\star$ s.t.\ $\Vrfy(\pp, \vrk{\id^\star},\msg^\star,\sigma^\star) \to \top$}
                    \STATE $\flag \ceq 1$
                \ENDIF
            \ENDIF
            \RETURN $\flag$
        \end{algorithmic}
    \end{algorithm}
    \vspace*{-1em}
    \caption{The unforgeability game for CMDVS.}
    \label{fig:CMDVS_UF}
\end{minipage}
%
\hfill
%
\begin{minipage}[t]{0.5\hsize}
    \centering
    \begin{algorithm}[H]
        \floatname{algorithm}{Experiment:}
        \caption{$\Exp{\CMDVS,\A}{\Cons}(\secpar)$}
        \label{exp:CMDVS_Consistency}
        \begin{algorithmic}[1]
            \STATE $(\pp,\sk)\gets\Setup(1^\secpar)$
            \STATE $(\Dv^\star,\msg^\star,\sigma^\star)\gets \A^{\KGOra,\SOra}(1^\secpar,\pp)$ 
            \IF{$\exists \id \in \Dv^\star \text{ s.t.\ } \Vrfy(\pp, \vrk{\id},\msg^\star,\sigma^\star) \to \top$}
                \IF{$\exists \id^\star \in \Dv^\star \text{ s.t.\ } \Vrfy(\pp, \vrk{\id^\star},\msg^\star,\sigma^\star) \to \bot$}
                    \STATE $\flag \ceq 1$
                \ENDIF
            \ENDIF
            \RETURN $\flag$
        \end{algorithmic}
    \end{algorithm}
    \vspace*{-1em}
    \caption{The consistency game for CMDVS.}
    \label{fig:CMDVS_Cons}
\end{minipage}
\end{figure}

\para{Unforgeability}
We define unforgeability as a standard security notion for CMDVS.
Intuitively, unforgeability guarantees that no adversary can (maliciously) modify a signature for $\Dv^\star\subset\VerS$ so that at least one non-designated verifier $\id\in\VerS\setminus\Dv^\star$ accepts it.
Specifically, we consider a security game, given in Fig.~\ref{fig:CMDVS_UF}, against an adversary $\A$, and let $\Adv{\CMDVS,\A}{\UF}(\secpar) \ceq \Pr[\Exp{\CMDVS,\A}{\UF}(\secpar)=1]$ be $\A$'s advantage in the game.

\begin{definition}[Unforgeability] \label{def:CMVDS_UF} {\rm
    Let $\CMDVS$ be a CMDVS scheme.
    $\CMDVS$ is said to meet unforgeability if for any sufficiently large $\secpar \in \NN$ and any PPTA $\A$, it holds $\Adv{\CMDVS,\A}{\UF}(\secpar) < \negl{\secpar}$.
}\end{definition}

\para{Consistency}
We consider \emph{consistency}, which was originally introduced by Damg\r{a}rd~et~al.~\cite{tccDHM+20} as a security notion for ordinary MDVS.
Roughly speaking, consistency guarantees that if at least one designated verifier accepts a signature, then all others also do so.
This notion is important in our setting, i.e., remote-control systems for IoT devices, for several possible reasons:
for example, it seems difficult to collect the acknowledgment messages from all IoT devices;
or, there might be only downstream communication from the systems manager to IoT devices.
Therefore, it seems hard to check which designated verifiers accepted a signature (without being maliciously modified).
Consistency allows the signer to just check a verification result of a specific designated verifier in order to confirm all verifiers accept the signature.%
\footnote{We assume all verifiers (including non-designated ones) receive the same data regardless of whether it is modified.}

Specifically, we consider a security game, given in  Fig.~\ref{fig:CMDVS_Cons}, against an adversary $\A$, and let $\Adv{\CMDVS,\A}{\Cons}(\secpar) \ceq \Pr[\Exp{\CMDVS,\A}{\Cons}(\secpar)=1]$ be $\A$'s advantage in the game.

\begin{definition}[Consistency] \label{def:CMVDS_Cons} {\rm
    Let $\CMDVS$ be a CMDVS scheme.
    $\CMDVS$ is said to meet consistency if for any sufficiently large $\secpar \in \NN$ and any PPTA $\A$, it holds $\Adv{\CMDVS,\A}{\Cons}(\secpar) < \negl{\secpar}$.
}\end{definition}

\section{CMDVS Constructions} \label{sec:Construction}

\subsection{Building Blocks} \label{sec:BuildingBlocks}

\para{Digital Signatures}
A digital signature $\DS = (\SigGen, \AB \SigSign, \AB \SigVer)$ is defined as follows.

\begin{enumerate}
    \renewcommand{\labelenumi}{--}
    %
    \item $\SigGen(1^\secpar) \to (\sigk,\verk)$:
    it takes a security parameter $\secpar$ as input and outputs a pair of a signing key and verification key $(\sigk,\AB \verk)$.
    \item $\SigSign(\sigk,\msg) \to \sigma$:
    it takes a signing key $\sigk$ and a message $\msg\in\M$ as input and outputs a signature $\sigma$.
    \item $\SigVer(\verk,(\msg,\sigma)) \to \top / \bot$:
    it takes a verification key $\verk$ and a pair of a message and a signature $(\msg,\sigma)$ as input and outputs $\top$ or $\bot$.
\end{enumerate}

\begin{definition}[Correctness] \label{def:DS_Correctness} {\rm
    Let $\DS$ be a digital signature scheme.
    For all $\secpar\in\NN$ all $(\sigk,\verk)\gets\SigGen(1^\secpar)$, all $\msg\in\M$, $\SigVer(\verk,(\msg,\SigSign(\sigk,\msg))) = \top$ holds with overwhelming probability.
}\end{definition}

\begin{figure}
    \centering
    \begin{algorithm}[H]
        \floatname{algorithm}{Experiment:}
        \caption{$\Exp{\DS,\A}{\CMA}(\secpar)$}
        \label{exp:DS_UF-CMA}
        \begin{algorithmic}[1]
            \STATE $(\sigk,\verk)\gets\SigGen(1^\secpar)$
            \STATE $(\msg^\star,\sigma^\star)\gets \A^{\SignOra{\cdot}}(1^\secpar,\verk)$
            \IF{$\msg^\star\notin\sM \wedge \SigVer(\verk,\msg^\star,\sigma^\star) \to \top$}
                \STATE $\flag \ceq 1$
            \ENDIF
            \RETURN $\flag$
        \end{algorithmic}
    \end{algorithm}
    \vspace*{-1em}
    \caption{ A $\UFCMA$ game for a digital signature $\DS$. $\SignOra{\cdot}$ is a signing oracle that returns $\SigSign(\sigk,\msg)$ for any query $\msg\in\M$ and adds $\msg$ to $\sM$.}
    \label{fig:DS_UF-CMA}
\end{figure}

A standard security notion for digital signatures is defined by a $\UFCMA$ game against a PPTA $\A$ in Fig.~\ref{fig:DS_UF-CMA}.

\begin{definition}[$\UFCMA$] \label{def:DS_UF-CMA} {\rm
    Let $\DS$ be a digital signature scheme.
    $\DS$ is said to be $\UFCMA$ secure if for sufficiently large $\secpar \in \NN$ and any PPTA $\A$, it holds $\Adv{\DS,\A}{\CMA}(\secpar) \ceq \Pr[\Exp{\DS,\A}{\CMA}(\secpar)=1] < \negl{\secpar}$.
}\end{definition}

\mpara{Approximate Membership Query (AMQ) Structures}
For an arbitrary set $\U \subset \bit{*}$, an AMQ data structure $\AMQ = (\Gen, \AB \Insert, \AB \Lookup)$ over $\U$ is defined as follows.%
\footnote{Although there are various AMQ structures supporting deletion operations, we do not consider them since we do not require deletion operations for our schemes.}

\begin{enumerate}
    \renewcommand{\labelenumi}{--}
    %
    \item $\Gen(\U,\prm) \to (\T,\aux)$:
    it takes $\U$ and a parameter $\prm$ as input, and outputs an initial structure $\T$ and auxiliary information $\aux$.
    The parameter $\prm$ varies depending on concrete AMQ structure constructions.
    \item $\Insert(\T,x,\aux) \to \T'$:
    it takes a data structure $\T$, an element $x\in\U$, auxiliary information $\aux$ as input, and outputs an updated structure $\T'$.
    \item $\Lookup(\T,x,\aux) \to \true / \false$:
    it takes a data structure $\T$, an element $x\in\U$, auxiliary information $\aux$ as input, and outputs $\true$ or $\false$.
\end{enumerate}

An AMQ structure meets the following completeness, while it allows false positives to make the structure size smaller and its probability can be bounded.
Note that false negatives never occur.

\begin{definition}[Completeness] \label{def:AMQ_Completeness} {\rm
    Let $\AMQ$ be an AMQ sturcture over $\U$.
    For any $\prm$, any $(\T_0,\aux)\gets \Gen(\U,\prm)$, any $\PS=\set{x_1, \AB \ldots, \AB x_{|\PS|}}\subset\U$, we define $\hT\ \ceq \T_{|\PS|}$ as $\T_i \gets \Insert(\T_{i-1},x_i,\aux)$ for $i \in [|\PS|]$.
   Then, for all $x\in\PS$, it holds $\Pr[\Lookup(\hT,x,\aux)=\true]=1$.
}\end{definition}

\begin{definition}[Bounded False-Positive Probability] \label{def:AMQ_FalsePositive} {\rm
    Let $\AMQ$ be an AMQ structure over $\U$, and suppose that $\hT$ is generated as in Def.~\ref{def:AMQ_Completeness} and $n \ceq |\PS|$.
    Then, there exists $\fp{n} \in (0,1]$ such that it holds $\Pr[\Lookup(\hT,x,\aux)=\true] \le \fp{n}$ for any $x\in\U\setminus\PS$, where the probability is over $\Gen$ and $\Insert$.
}\end{definition}

AMQ structures mainly aim to compress the description length of $\PS$ by allowing false positive errors.
Therefore, the size of the structure $\hT$ should be smaller than the following trivial solutions:
(1) encode each element of $\PS$ and list them, i.e., $|\PS| \cdot \log_2 |\U|$ bits; and
(2) prepare an $|\U|$-bit string and set every $i$-th bit to one if and only if $x_i \in \PS$.
Namely, it should hold $|\hT| \le \min\{|\PS| \cdot \log_2 |\U|, |\U|\}$.

There are many instantiations of AMQ structures: the Bloom filter~\cite{comacmBlo70} and its variants~\cite{sodaPPR05,esaKM06}, cuckoo filter~\cite{conextFAKM14}, vacuum filter~\cite{vldbWZSQ19}, etc.
Although the Bloom filter has been theoretically well-analyzed due to its simple structure, recent constructions (e.g.,~\cite{conextFAKM14,vldbWZSQ19}) are (experimentally) more efficient in terms of structure sizes.
In Section~\ref{sec:InstantiationAndPerformance}, we will give formal description of the Bloom filter.

\subsection{Trivial Construction} \label{sec:Construction_Trivial}
A digital signature scheme $\DS = \lrparen{\SigGen,\SigSign,\SigVer}$ can be used to trivially construct a CMDVS scheme $\CMDVSall$ as follows.

\begin{enumerate}
    \renewcommand{\labelenumi}{--}
    %
    \item $\Setup(1^{\secpar})$:
    It runs $(\sigk,\verk) \gets \SigGen(1^\secpar)$ and returns $(\pp,\sk)$, where $\pp \ceq \verk$ and $\sk \ceq \sigk$.
    \item $\KeyGen(\pp,\sk,\id)$:
    It returns $\vrk{\id} \ceq \verk$.
    \item $\Sign(\sk,\Dv,\msg,\len)$:
    It runs $\sigmaDS \gets \SigSign(\sigk,\Dv\|\msg)$ and sets $\sigma \ceq (\Dv,\sigmaDS)$. It returns $\bot$ if $|\sigma|>\len$; it returns $\sigma$ otherwise.
    \item $\Vrfy(\pp,\vrk{\id},\msg,\sigma)$:
    If $\id\notin \Dv$, it returns $\bot$. Otherwise, it returns the output of $\SigVer(\verk,\lrparen{\Dv\|\msg, \AB \sigmaDS})$.
\end{enumerate}

The above construction clearly meets the correctness, unforgeability, and consistency. 
We omit the proof.

\begin{theorem} \label{thm:CMDVS_Trivial} {\rm
    If $\DS$ meets $\UFCMA$ security, the above CMDVS scheme $\CMDVS$ meets correctness, unforgeability, and consistency.
}\end{theorem}

Although the above construction is quite simple, the signature size $|\sigma|$ is $|\Dv|\cdot \log_2 |\ID|+|\sigmaDS|$.
Namely, the maximum signature length $\len$ must always satisfy $\len \ge |\Dv|\cdot \log_2 |\ID|+|\sigmaDS|$.

Out construction in the next section realizes smaller signature sizes;
in particular, it can flexibly specify $\len$ s.t.\ $\len = o(|\Dv|)$ with adjustment for other parameters (see Remark~\ref{rem:CMDVS_ParameterSetting} for details).

\subsection{Proposed Generic Construction} \label{sec:Construction_Generic}
We show a CMDVS scheme from an AMQ structure and DS scheme.
Compared to the trivial construction, we can succeed in drastically reducing the signature size by allowing a \emph{small} false-positive probability, which can be made negligible with appropriate parameter settings.

In the following, we suppose a function $\Assign: \NN \times \ID \to 2^{\U}$ over $\U$.
Roughly speaking, $\Assign$ is a function that uniquely assigns multiple elements in $\U$ to an arbitrary identity, and we assume that for any fixed $\ell\in\NN$ and for any $\id,\id'\in\ID$, it holds $\Assign(\ell,\id)\cap\Assign(\ell,\id')=\emptyset$.
Note that such a function can be realized in the following way:
suppose $\ID \ceq \bit{\gamma},\U \ceq \bit{\gamma+\lrfloor{\log_2 \ell}+1}$, and for any $\ell$ and any $\id \in \ID$, we define $\Assign(\ell,\id) \ceq \set{\beta_1\|\id, \AB \beta_2\|\id, \AB \ldots, \AB \beta_\ell\|\id}$, where $\beta_i$ is binary representation of $i \in [\ell]$.
Our CMDVS scheme from an AMQ structure $\AMQall$ over $\U \subset \bit{*}$ and a DS scheme $\DSall$ as follows.

\begin{enumerate}
    \renewcommand{\labelenumi}{--}
    %
    \item $\Setup(1^{\secpar})$:
    It arbitrarily chooses $\ell \in \NN$, and it returns $(\pp,\sk)$, where $\pp \ceq (\verk, \AB \ell)$ and $\sk \ceq (\sigk, \ell)$.
    \item $\KeyGen(\pp,\sk,\id)$:
    It returns $\vrk{\id} \ceq \Assign(\ell,\id)$.
    \item $\Sign(\sk,\Dv,\msg,\len)$:
    It derives an appropriate parameter $\prm$ from $\Dv$, $\msg$, and $\len$.
    If $\prm$ cannot be derived, it returns $\bot$.
    For every $\id_i \in \Dv$, let $\X_{i} = \set{x_{(i-1)\ell+1}, \AB \ldots, \AB x_{i\ell}} \AB \ceq \AB \Assign(\ell, \AB \id_i)$.%
    \footnote{Namely, $\bigcup_{i=1}^{|\Dv|} \X_i = \set{x_1, x_2, \ldots, x_{\ell |\Dv|}}$.}
    It runs $(\T_0,\aux)\gets \Gen(\U,\prm)$ and for every $i \in [\ell |\Dv|]$, it computes $\hT \ceq \T_{\ell |\Dv|}$ as follows:
   \begin{align*}
        \T_i \gets \Insert(\T_{i-1},x_{i}, \aux).
   \end{align*}
    It sets $\sigma \ceq (\hT, \AB \aux, \AB \sigmaDS)$, where $\sigmaDS \gets \SigSign(\sigk, \AB \msg\|\hT\|\aux)$.
    If $|\sigma|>\len$, it returns $\bot$; otherwise, it returns $\sigma$.
    \item $\Vrfy(\pp,\vrk{\id},\msg,\sigma)$:
    It runs $\SigVer(\verk,\lrparen{\msg\|\hT\|\aux,\sigmaDS})$. 
    If the output is $\bot$, it returns $\bot$ and terminates.
    For every $x \in \X_{\id}$, it returns $\bot$ and terminates if $\Lookup(\hT, \AB x, \AB \aux)$ outputs $\false$.
    It returns $\top$ (if all $\Lookup$ outputs are $\true$).
\end{enumerate}

The above construction meets the desirable properties below.

\begin{theorem} \label{thm:CDMVS_Generic} {\rm
    If a DS scheme $\DS$ meets $\UFCMA$ security and an AMQ structure $\AMQ$ meets completeness and bounded false-positive probability such that it holds $\fp{\ell|\Dv|}=2^{-\bigO{\secpar}}$ for all possible $\ell \in \NN$ and $\Dv\subset\VerS$ in the above construction, the above CMDVS scheme $\CMDVS$ meets correctness, unforgeability, and consistency.
}\end{theorem}
\begin{proof}
We prove the correctness, unforgeability, and consistency of $\CMDVS$ as follows.

\mpara{Correctness}
Roughly speaking, we can prove the correctness property of $\CMDVS$ from completeness and bounded false-positive probability of $\AMQ$ and correctness of $\DS$.

We fix an arbitrary subset $\VerS\subset\ID$, and let $\vrk{\id} \ceq \Assign(\ell,\id)$ for every $\id \in \VerS$.
We also fix an arbitrary subset $\Dv \subset \VerS$.
\footnote{Although $\Dv \subset \ID$ is considered in Def.~\ref{def:CMDVS_Correctness}, we here consider the case of $\Dv \subset \VerS$ without loss of generality.}
For any $\msg \in \M$ and any $\len \in \NN$, $\Sign$ outputs $\sigma \ceq \lrparen{\hT, \AB \aux, \AB \sigmaDS}$.
Due to the correctness property of $\DS$ (Def.~\ref{def:DS_Correctness}), it clearly holds that $\SigVer(\verk,\lrparen{\msg\|\hT\|\aux,\sigmaDS})=\top$.
We then consider the following two cases: $\id \in \Dv$ and $\id \in \VerS \setminus \Dv$.
\begin{description}
    \item[The case of $\id \in \Dv$.]
    Due to the completeness property of $\AMQ$ (Def.~\ref{def:AMQ_Completeness}), for any $\id \in \Dv$ and any $x\in\X_{\id}$, it is obvious that it holds $\Lookup(\hT,x,\aux) = \true$.
    Therefore, $\Vrfy(\pp,\vrk{\id},\msg,\sigma)$ always outputs $\top$.
    \item[The case of $\id \in \VerS\setminus\Dv$.]
    Due to bounded false-positive probability of $\AMQ$ (Def.~\ref{def:AMQ_FalsePositive}), for any $\id \in \VerS\setminus\Dv$ and any $x\in\X_{\id}$, it holds that $\Pr[\Lookup(\hT,x,\aux) = \true] \le \fp{\ell|\Dv|}$.
    We then have
    \begin{align*}
        &      \Pr\LRbra{\Vrfy(\pp,\vrk{\id},\msg,\sigma)=\top} \\
        &\qquad  \le \LRparen{\fp{\ell|\Dv|}}^{\ell} 
           =   2^{-\ell \cdot \bigO{\secpar}} 
           =   \negl{\secpar},
    \end{align*}
    where $\negl{\secpar}$ is a negligible function.
\end{description}

\mpara{Unforgeability}
Loosely speaking, $\UFCMA$ security of $\DS$ guarantees unforgeability unless $(\hT, \AB \aux)=(\hT', \AB \aux')$ occurs for distinct $\Dv,\Dv'\subset\VerS$, where $\sigma \ceq (\hT, \AB \aux, \AB \sigmaDS) \gets \Sign(\sk, \AB \Dv, \AB \msg, \AB \len)$ and $\sigma' = (\hT', \AB \aux', \AB \sigmaDS') \gets \Sign(\sk, \AB \Dv', \AB \msg, \AB \len)$.
If it occurs, $\sigma$, which is a valid signature for $(\msg,\Dv)$, is also a valid one for $(\msg,\Dv')$;
it breaks unforgeability.
The following lemma shows such a situation occurs with negligible probability.

\begin{lemma} \label{lem:Proof_Unforgeability_Conflict} {\rm
    Let $\CMDVS$ be a CMDVS scheme and $\AMQ$ be an AMQ structure with completeness and bounded false-positive probability such that it holds $\fp{\ell|\Dv|}=2^{-\bigO{\secpar}}$ for any $\ell\in\NN$ and $\Dv\subset\VerS$.
    Then, for any $\msg,\msg' \in \M$, any $\len, \len' \in \NN$, any $\VerS \subset \ID$, and any distinct $\Dv,\Dv' \subset \VerS$, it holds
    \begin{align*}
        \Pr\LRbra{(\hT, \AB \aux)=(\hT', \AB \aux')} \le \negl{\secpar},
    \end{align*}
    where $(\hT, \AB \aux, \AB \sigmaDS) \gets \Sign(\sk, \AB \Dv, \AB \msg, \AB \len)$ and $(\hT', \AB \aux', \AB \sigmaDS') \gets \Sign(\sk, \AB \Dv, \AB \msg', \AB \len')$.
}\end{lemma}
\begin{proof}[Proof of Lemma~\ref{lem:Proof_Unforgeability_Conflict}]
We assume that for some $\msg,\msg' \in \M$, some $\len, \len' \in \NN$, some $\VerS \subset \ID$, and some distinct $\Dv,\Dv' \subset \VerS$, it holds $(\hT, \AB \aux)=(\hT', \AB \aux')$ with non-negligible probability, where $(\hT, \AB \aux, \AB \sigmaDS) \gets \Sign(\sk, \AB \Dv, \AB \msg, \AB \len)$ and $(\hT', \AB \aux', \AB \sigmaDS') \gets \Sign(\sk, \AB \Dv, \AB \msg', \AB \len')$.
We show a contradiction.
Since $\Dv \neq \Dv'$, there exists $\id^\star \in \Dv'\setminus\Dv$ or $\id^\star \in \Dv\setminus\Dv'$.
Without loss of generality, suppose $\id^\star \in \Dv'\setminus\Dv$.
Let $\vrk{\id^\star} = \X_{\id^\star}$.
By the assumption and the completeness of $\AMQ$, for any $x \in \X_{\id}$, we have $\Lookup(\hT, x, \aux) = \Lookup(\hT', x, \aux') = \true$.
This means that for $\Dv$, a false positive occurs with non-negligible probability, which contradicts bounded false-positive probability for $\Dv$, which should be negligible, i.e., $\fp{\ell|\Dv|} \le \negl{\secpar}$.
\end{proof}

Thus, we can easily show that if there exists a PPTA $\A$ that breaks the unforgeability of $\CMDVS$, there exists a PPTA $\mathcal{F}$ that breaks $\UFCMA$ security of $\DS$.
We omit the proof since it is straightforward.

\mpara{Consistency}
It clearly follows from completeness and bounded false-positive probability of $\AMQ$ and $\UFCMA$ security of $\DS$.
Roughly speaking, Lemma~\ref{lem:Proof_Unforgeability_Conflict}, which requires bounded false-positive probability of $\AMQ$, guarantees that $(\hT,\aux)$ is a uniquely determined by a designated-verifier set $\Dv$.
Namely, there exists at most one $(\hT,\aux)$ per $\Dv$.
$\UFCMA$ security of $\DS$ guarantees that for any $\sigma = (\hT,\aux,\sigmaDS)$, $(\hT,\aux)$ is correctly generated by the signer as long as $\sigmaDS$ is valid.
Finally, completeness of $\AMQ$ guarantees that all designated verifiers $\Dv$ accept correctly-generated signatures $\sigma$.

\medskip
It completes the proof.
\end{proof}

\mpara{Instantiations}
The above construction can be instantiated with any AMQ structures and digital signatures.
After the seminal work of AMQ structures, i.e., the Bloom filter~\cite{comacmBlo70}, there are various (heuristically) efficient AMQ structure constructions such as the cuckoo filter~\cite{conextFAKM14} and the vacuum filter~\cite{vldbWZSQ19}.
In this paper, we will employ the Bloom filter and the vacuum filter as the underlying AMQ structures for theoretical performance analysis in Section~\ref{sec:Analysis} and implementations in Section~\ref{sec:experiments}, respectively.

\subsection{System Description} \label{sec:Description}
We give a concrete description of $\system$ with CMDVS $\CMDVSall$.
We consider a message space $\M$ of $\CMDVS$ as a command space for $\system$.

\spara{System Setup}
$\SysM$ runs $\Setup$ with an appropriate security parameter $\secpar$ to get a public parameter $\pp$ and a signing key $\sk$.
$\SysM$ updates a identifier set (or a list) $\IDact$ of activated IoT devices.

\spara{Embedding Keys}
For any device $\IoT{\id}$, $\SysM$ runs $\KeyGen(\pp,\sk,\id)$ and obtains $\vrk{\id}$.
$\SysM$ then embeds or sends $(\pp,\vrk{\id})$ into the device $\IoT{\id}$.


\spara{Sending Requests}
$\Owner$ sends $\SysM$ a request to have an arbitrary subset $\set{\IoT{\id}}_{\id \in \IDdsg}$ of activated IoT devices execute a command $\cmd \in \M$.
Namely, the request includes $\IDdsg$ and $\cmd$.
Note that they can securely communicate with each other using the SSL/TLS.

\spara{Broadcast}
$\SysM$ runs $\Sign(\sk,\IDdsg,\cmd,\len)$ to obtain $\sigma$, where $\len$ may be set depending on environment, i.e., it might be set at the beginning of the system or every broadcast, etc.
$\SysM$ then broadcasts an authenticated command $\acmd \ceq (\cmd,\sigma)$ to all devices.

\spara{Command Verification}
Suppose every $\IoT{\id}$ receives an authenticated command $\acmd'$ and parse $\acmd' = (\cmd',\sigma')$.
It then runs $\Vrfy(\pp,\vrk{\id},\cmd',\sigma')$.
If it outputs $\top$, then $\IoT{\id}$ confirms that $\id$ was designated and $\cmd'$ is a valid one, and executes $\cmd'$.
Otherwise, $\IoT{\id}$ does nothing.


\medskip
It is obvious that the above system meets completeness and integrity from correctness, unforgeability, and consistency of the underlying CMDVS $\CMDVS$.
Thanks to the underlying AMQ structures, our CMDVS construction can achieve constant-size signatures by appropriately setting parameters, and it also provides efficient verification since it only requires a single signature verification and $\ell$ lookup operations.
Note that lookup operations are basically lightweight;
for example, Bloom filter's lookup operation is constructed with only (non-cryptographic) hash functions.
Hence, the above system clearly meets scalability and light weight.

\section{Concrete Instantiation} \label{sec:InstantiationAndPerformance}
To clarify the effectiveness of our proposed construction, we instantiate the underlying AMQ structure and show an instantiation of our construction from (an improved variant of) the Bloom filter and any digital signatures.

\subsection{An Improved Variant of Bloom Filter} \label{sec:Bloom}
We describe the Bloom filter employed in our instantiation.
Roughly speaking, we employ Kirsch and Mitzenmacher's technique~\cite{esaKM06} to simplify the traditional construction~\cite{comacmBlo70} of the Bloom filter.
Their technique reduces the number of hash functions used in the Bloom filter construction, and effectively implements the Bloom filter without any increase in the asymptotic false-positive probability.

\mpara{Parameters}
In the Bloom filter, $\prm$, which is input of $\Gen$, consists of the following four parameters.
\begin{itemize}
    \setlength{\itemsep}{0em}
    \setlength{\parskip}{0em}
    \item $m$: the size of data structure $\T$. Namely, we have $|\T| = m$.
    \item $n$: the (maximum) number of elements inserted to $\T$.
    \item $\mu$: an upper bound of false-positive probability. Namely, it holds $\fp{|\PS|} \le \mu$ for any $\PS\subset\U$.
    \item $k$: the number of hash functions used in Bloom filter.
\end{itemize}
We will discuss how to determine $m$, $n$, $\mu$, and $k$ later.

\mpara{Construction}
We employ (non-cryptographic) hash functions such as FNV-1a\footnote{\url{http://www.isthe.com/chongo/tech/comp/fnv/}} and Murmur2.\footnote{\url{https://github.com/aappleby/smhasher/blob/master/src/MurmurHash2.cpp}}
(A variant of) the Bloom filter $\Bloom = \lrparen{\Gen, \AB \Insert, \AB \Lookup}$ is constructed as follows.

\begin{enumerate}
    \renewcommand{\labelenumi}{--}
    %
    \item $\Gen(\U,\prm) \to (\T,\aux)$:
    For every $i\in\bin$, it randomly chooses hash functions $\hash_i: \U \to [m]$.
    It returns $(\T,\aux)$, where $\T \ceq 0^m$ and $\aux \ceq (k,\hash_0,\hash_1)$.
    \item $\Insert(\T,x,\aux) \to \T'$:
    For every $i\in[k]$, it computes $\Hash_i(x) \ceq \hash_0(x) + i \cdot \hash_1(x) \bmod{m}$ and $\T[\Hash_i(x) + 1] \ceq 1$.%
    \footnote{Since $\Hash_i(x) \in [0,m-1]$, we need to set $\T[\Hash_i(x) + 1] \ceq 1$, not $\T[\Hash_i(x)] \ceq 1$.}
    It returns $\T' \ceq \T$.
    \item $\Lookup(\T,x,\aux) \to \true / \false$:
    For every $i\in[k]$, it returns $\false$ and terminates the process if $\T[\Hash_i(x)+1] = 0$, where $\Hash_i(x) \ceq \hash_0(x) + i \cdot \hash_1(x) \bmod{m}$.
    It returns $\accept$ (if $\T[\Hash_i(x)+1] = 1$ for all $i\in[k]$).
\end{enumerate}

Goel and Gupta~\cite{sigmetricsGG10} showed the following lemma.
\begin{lemma}[\cite{sigmetricsGG10}] \label{lem:Bloom_Relation} {\rm
    Let $\Bloom = \lrparen{\Gen, \AB \Insert, \AB \Lookup}$ be the Bloom filter.
    For any $(m,n,k) \in [|\U|]^2 \times\NN$ and any $q \in [|\U|-n]$, let
    \begin{align}
        \mu \ceq \LRparen{1-e^{-\frac{\LRparen{n+(q/2)}k}{m-1}}}^{kq}, \label{eq:mu}
    \end{align}
    and $\prm \ceq (m,n,\mu,k)$.
    Then, for any $(\T_0,\aux) \gets \Gen(\U,\prm)$ and any $\PS = \set{x_1, \AB \ldots, \AB x_{n}} \subset \U$ such that $|\PS|=n$, we define $\hT \ceq \T_{n}$, where $\T_i \gets \Insert(\T_{i-1},x_i,\aux) \text{ for } i \in [n]$.
    We say that the false positive occurs if for any $\Q \subset \U\setminus\PS$ such that $|\Q|=q$, it holds $\Lookup(\hT,x,\aux)=\accept$ for all $x\in\Q$. Then, the false-positive probability $p$ satisfies $p\le\mu$.
}\end{lemma}
Note that the above lemma includes $\fp{n}$ in Def.~\ref{def:AMQ_FalsePositive} as a special case when we set $q=1$.
\begin{corollary}[\cite{sigmetricsGG10}] \label{cor:Bloom_Relation} {\rm
    Let $\Bloom = \lrparen{\Gen, \AB \Insert, \AB \Lookup}$ be the Bloom filter.
    For any $(m,n,k) \in [|\U|]^2 \times\NN$, we set
     \begin{align*}
        \mu \ceq \LRparen{1-e^{-\frac{\LRparen{n+(1/2)}k}{m-1}}}^{k}.
    \end{align*}  
    Suppose that $\hT$ is generated as in Lemma~\ref{lem:Bloom_Relation}.
    Then, $\fp{n}$ defined in Def.~\ref{def:AMQ_FalsePositive} satisfies $\fp{n} \le \mu$.
}\end{corollary}

\subsection{Instantiation from the Bloom Filter and Any Digital Signatures} \label{sec:Instantiation}

We show a concrete instantiation of the proposed CMDVS construction $\CMDVSall$ in Section~\ref{sec:Construction_Generic} with the parameter-tuned Bloom filter.
The most crucial part is parameter adjustment for the Bloom filter.

\mpara{Parameter setting}
Based on Lemma~\ref{lem:Bloom_Relation}, we can flexibly set the Bloom filter parameters for our CMDVS construction as follows.

    %
As can be seen in the proposed construction, $\Setup$ determines $\ell$ at Step~2.
In this instantiation, for the security parameter $\secpar \in \NN$, $\Setup$ can choose arbitrary $\ell \in [\secpar]$.

$\Sign$ derives $\prm = (m,n,\mu,k)$ from $\Dv$, $\msg$, and $\len$, at Step~1.
In this instantiation, $\prm$ is determined as follows.
By setting $n \ceq \ell |\Dv|$ and $q \ceq \ell$ in Lemma~\ref{lem:Bloom_Relation}, Eq.~\eqref{eq:mu} can be written as:
\begin{align}
    \mu \ceq \LRparen{1-e^{-\frac{\LRparen{|\Dv|+(1/2)}k\ell}{m-1}}}^{k\ell}. \label{eq:CMDVS_FP_UpperBound}
\end{align}
Now, we would like to set $k$ and $m$ so that they satisfy
\begin{align}
    \mu \le \frac{1}{2^{c\secpar}}, \label{eq:CMDVS_secpar}
\end{align}
for arbitrary constant $c \in \RR$, since the proposed construction requires the AMQ structure $\AMQ$ with the negligible false-positive probability (see Theorem~\ref{thm:CDMVS_Generic}).
To achieve that aim, let $K \ceq k \ell$ for convenience and we set $k$ such that 
\begin{align}
    k \ell   =    \LRfloor{\frac{(m-1)\ln 2}{|\Dv|+1/2}}. \label{eq:CMDVS_K}
\end{align}    
Then, from Eq.~\eqref{eq:CMDVS_FP_UpperBound}, we have
\begin{align}
    \!\!\!\!
    \mu &   =   \LRparen{1-e^{-\frac{\LRparen{|\Dv|+(1/2)}K}{m-1}}}^{K} \!\!
           \le \LRparen{1-e^{-\frac{\LRparen{|\Dv|+(1/2)}\frac{(m-1)\ln 2}{|\Dv|+1/2}}{m-1}}}^{K} \!\!
           \le \frac{1}{2^K}.   \label{eq:CMDVS_mu_K}
\end{align}
From Eqs.~\eqref{eq:CMDVS_secpar}~and~\eqref{eq:CMDVS_mu_K}, $\mu$ has to satisfy $\mu \le 1/2^K \le 1/2^{c \secpar}$.
Therefore, we have to set $(k,\ell)$ so that it satisfies
\begin{align}
    k \ell = K \ge c \secpar. \label{eq:CMDVS_K_relation}
\end{align}
Since $\ell$ was already chosen by $\Setup$, $k \ceq \LRceil{c \secpar/\ell}$ satisfies Eq.~\eqref{eq:CMDVS_K_relation}.
Now we are ready to choose $m$.
From Eqs.~\eqref{eq:CMDVS_K}--\eqref{eq:CMDVS_K_relation}, we have
\begin{align}
    \frac{(m-1)\ln 2}{|\Dv|+1/2} \ge k \ell \ge c \secpar. \label{eq:CMDVS_K_secpar}
\end{align}    
Namely, $m$ has to satisfy the following inequality to meet Eq.~\eqref{eq:CMDVS_K_secpar}:
\begin{align*}
    m   \ge \frac{(|\Dv|+1/2)k\ell}{\ln 2}+1.  
\end{align*}
Therefore, the following $m$ is sufficient:
\begin{align*}
    m   \ceq \LRceil{\frac{(|\Dv|+1/2)k\ell}{\ln 2}}+1.
\end{align*}
Note that the above parameters $(c,k,\mu,m,n)$ can be adaptively set every time $\Sign$ is executed.

\begin{remark}[Towards CMDVS with Constant-Size Signatures] \label{rem:CMDVS_ParameterSetting} {\rm
Although the above parameter setting works for any $n = \ell |\Dv|$, i.e., for any $\Dv \subset \VerS$, the Bloom filter also allows us to fix the size of data structures $m$ first and then determine concrete $n$, i.e., an upper bound of the size of $\Dv$.
Namely, the Bloom filter can also provide a concrete CMDVS construction with the constant-size signatures regardless of the size of $\Dv$, though $\Dv$ has to satisfy $|\Dv| \le n$, where $n$ is fixed throughout the protocol and determined according to the constant $m$.
From Eq.~\eqref{eq:CMDVS_K_secpar}, for any constant $m \in \NN$, $|\Dv|$ has to satisfy the following inequality:
        \begin{align*}
            |\Dv|   \le \frac{(m-1)\ln 2}{k \ell}-\frac{1}{2}.
        \end{align*}
        Hence, $n$ has to satisfy the following:
        \begin{align*}
            n = \ell |\Dv|   
            &   \le \ell \LRparen{ \frac{(m-1)\ln 2}{k \ell}-\frac{1}{2}} 
               = \frac{(m-1)\ln 2}{k}-\frac{\ell}{2}.
        \end{align*}
        Thus, $n$ is defined as follows.
        \begin{align*}
            n   &   \ceq    \LRfloor{\frac{(m-1)\ln 2}{k}-\frac{\ell}{2}}.
        \end{align*}
        Note that $\Sign$ outputs $\bot$ when $|\Dv| > n$.
}\end{remark}

\mpara{Instantiation}
With the above parameters, we can instantiate the proposed construction in Section~\ref{sec:Construction_Generic} by the Bloom filter $\Bloomall$ over $\U \subset \bit{*}$.
To be precise, this concrete construction is a slightly-modified but more efficient version of an instantiation of $\CMDVS$ from the Bloom filter. Note that the modification does not affect the security proofs in Theorem~\ref{thm:CDMVS_Generic}. 
We add footnotes on the differences between the simple instantiation and ours in this section at where we make the modifications.

\begin{enumerate}
    \renewcommand{\labelenumi}{--}
    %
    \item $\Setup(1^{\secpar})$:
    Run $(\sigk,\verk) \gets \SigGen(1^\secpar)$.
    It arbitrarily chooses $\ell \in [\secpar]$ and randomly choose $\hash_i: \U \to [2^\secpar]$ for every $i\in\bin$, and it returns $(\pp,\sk)$, where $\sk \ceq (\sigk, \AB \ell, \AB \hash_0, \AB \hash_1)$, $\pp \ceq (\verk, \AB \ell)$.
    \item $\KeyGen(\pp,\sk,\id)$:
    For every $i\in[\ell]$, it computes $h_{\id,0}^{(i)} \ceq \hash_0(x_{\id}^{(i)})$ and $h_{\id,1}^{(i)} \ceq \hash_1(x_{\id}^{(i)})$, where $\X_{\id} =\set{x_{\id}^{(1)},\ldots,x_{\id}^{(\ell)}} \ceq \Assign_{\ell}(\id)$. It returns $\vrk{\id} \ceq \set{\lrparen{h_{\id,0}^{(i)},h_{\id,1}^{(i)}}}_{i\in[\ell]}$.
    \item $\Sign(\sk,\Dv,\msg,\len)$:
    It derives $(c,k,\mu,m,n)$ as above.
    If $\ell |\Dv| > n$ holds, it returns $\bot$.
    For every $\id_i \in \Dv$, let $\X_{i} = \set{x_{(i-1)\ell+1}, \AB \ldots, \AB x_{i\ell}} \AB \ceq \AB \Assign(\ell, \AB \id_i)$.%
    It initializes $\T_0$ as $\T_0 \ceq 0^m$ and computes $\hT \ceq \T_{\ell|\Dv|}$ as follows:
    \begin{align*}
        \T_i \gets \Insert(\T_{i-1},x_{i}, (\ell |\Dv|,k,\hash_0,\hash_1)) \text{ for } i \in [\ell |\Dv|].
    \end{align*}
    Namely, for every $i \in [\ell]$ and every $j\in[k]$, it computes $H_j^{(i)} \ceq h_{\id,0}^{(i)}+j \cdot h_{\id,1}^{(i)} \bmod{m}$ and sets $\hT[H_j^{(i)}+1] \ceq 1$.
    It sets $\sigma \ceq (\hT, \AB \aux, \AB \sigmaDS)$, where $\aux = k$ and $\sigmaDS \gets \SigSign(\sigk,\msg\|\hT\|\aux)$.%
    \footnote{$\aux$ only consists of $k$ since hash functions $\hash_0,\hash_1$ were already chosen by $\Setup$.}
    If $|\sigma|>\len$, it returns $\bot$; otherwise, it returns $\sigma$.
    \item $\Vrfy(\pp,\vrk{\id},\msg,\sigma)$:
    It runs $\SigVer(\verk,\lrparen{\msg\|\hT\|\aux,\sigmaDS})$. 
    If the output is $\bot$, it returns $\bot$ and terminates.
    Let $m \ceq |\hT|$. For every $i \in [\ell]$ and every $j\in[k]$, it returns $\bot$ and terminates if $\hT[H_j^{(i)}+1] = 0$, where $H_j^{(i)} \ceq h_{\id,0}^{(i)}+j \cdot h_{\id,1}^{(i)} \bmod{m}$.%
    \footnote{This step corresponds to $k \ell$ executions of $\Lookup$, though AMQ elements $x_{\id}^{(1)},\ldots,x_{\id}^{(\ell)}$ were already embedded into hash functions when generating $\vrk{\id}$.}
    It returns $\top$ (if $\hT[H_j^{(i)}+1] = 1$ for all $i \in [\ell]$ and all $j\in[k]$).
\end{enumerate}

\begin{table}[t]
\caption{
Efficiency comparison between our instantiation and the trivial construction.
}
\label{tab:EfficiencyComparison}
%
\vspace*{-1em}
\begin{center}
\begin{tabular}{ccccc}
\hline
\multirow{2}{*}{\centering  $|\Dv|$}	&   \multicolumn{3}{c}{Instantiated (\S \ref{sec:Instantiation})}	&   Trivial (\S \ref{sec:Construction_Trivial})\\
\cmidrule(lr){2-4}
\cmidrule(lr){5-5}
&	$\mu=2^{-c\kappa}$	&	$k$    &	$|\sigma|$ (bits)	&	$|\sigma|$ (bits)	\\[1pt]
\hline\hline
\multirow{3}{*}{\centering  $100$}  &	$2^{-10}$	&	$\LRceil{10/\ell}$	&	$1{,}967$	&	\multirow{3}{*}{\centering  $6{,}912$}	\\
                                    &	$2^{-15}$	&	$\LRceil{15/\ell}$	&	$2{,}692$	&	 	\\
                                    &	$2^{-20}$	&	$\LRceil{20/\ell}$	&	$3{,}418$	&		\\
\hline
\multirow{3}{*}{\centering  $1{,}000$}  &	$2^{-10}$	&	$\LRceil{10/\ell}$	&	$14{,}952$	&	\multirow{3}{*}{\centering  $64{,}512$}	\\
                                        &	$2^{-15}$	&	$\LRceil{15/\ell}$	&	$22{,}169$	&	 	\\
                                        &	$2^{-20}$	&	$\LRceil{20/\ell}$	&	$29{,}387$	&		\\
\hline
\multirow{3}{*}{\centering  $10{,}000$} &	$2^{-10}$	&	$\LRceil{10/\ell}$	&	$144{,}794$	&	\multirow{3}{*}{\centering  $640{,}512$}	\\ 
                                        &	$2^{-15}$	&	$\LRceil{15/\ell}$	&	$216{,}933$	&	 	\\
                                        &	$2^{-20}$	&	$\LRceil{20/\ell}$	&	$289{,}052$	&		\\
\hline
\multirow{3}{*}{\centering  $100{,}000$}    &	$2^{-10}$	&	$\LRceil{10/\ell}$	&	$1{,}443{,}220$	&	\multirow{3}{*}{\centering  $6{,}400{,}512$}	\\
                                            &	$2^{-15}$	&	$\LRceil{15/\ell}$	&	$2{,}164{,}571$	&	 	\\
                                            &	$2^{-20}$	&	$\LRceil{20/\ell}$	&	$2{,}885{,}923$	&		\\
\hline
\multirow{3}{*}{\centering  $1{,}000{,}000$}    &	$2^{-10}$	&	$\LRceil{10/\ell}$	&	$14{,}427{,}475$	&	\multirow{3}{*}{\centering  $64{,}000{,}512$}	\\
                                                &	$2^{-15}$	&	$\LRceil{15/\ell}$	&	$21{,}640{,}954$	&	 	\\
                                                &	$2^{-20}$	&	$\LRceil{20/\ell}$	&	$28{,}854{,}434$	&		\\
\hline
\end{tabular}
\end{center}
\end{table}

\subsection{Theoretical Analysis} \label{sec:Analysis}
We give an efficiency comparison between the concrete instantiation in the previous section and the trivial construction in Section~\ref{sec:Construction_Trivial}.

We set (the upper bound of) the false-positive probability $\mu$ in the range of $2^{-10}$ to $2^{-20}$, which is comparable to the false-positive probability of our experimental results in the next section.
Note that the false-positive probability is related to correctness;
each non-designated verifier rejects correctly-generated signatures with probability $1-\mu$ in our instantiation.

There is a trade-off between $\ell$ and the number of hash functions $k$.
Therefore, small $\ell$ makes the sizes of verification keys compact, whereas the computational costs, which depend on $k$, increase.

We assume that the bit length of $\id$ is 64 bits and the DS signature size is 512 bits (assuming the EdDSA signatures as in the experiment section).
Therefore, the signature sizes are calculated by $\lrfloor{(|\Dv|+1/2)k\ell/\ln 2}+\lrfloor{\log_2 k}+514$ for the instantiation and $|\Dv|\cdot 64+512$ for the trivial construction.
Obviously, compared to the trivial construction, our instantiation enables a $40$\%--$65$\% size reduction of the signature size, depending on the false-positive probability $\mu$.
Although we employed the Bloom filter since we wanted to see the theoretical performance of the proposed construction, in the next section, we implement $\system$ from our generic construction instantiated with the vacuum filter~\cite{vldbWZSQ19}, which is a more recent AMQ structure yielding experimentally better performance.

\section{Experiments} \label{sec:experiments}

In this section, we describe experimental evaluations of $\system$. 
Our primary motivation for the evaluations is to confirm how communication sizes can be reduced by virtue of an AMQ structure compared with the trivial construction and broadcast authentication~\cite{ainaWYS21}, which supports the functionality that a sender chooses an arbitrary subset of receivers.\footnote{Although the broadcast authentication in~\cite{ainaWYS21} is based on message authentication codes (MAC), we simply say signatures as MAC for the sake of convenience.}

We first describe our implementation of the proposed CMDVS constructions through their instantiations and then demonstrate experimental results, including the computation time on a laptop PC. 
Finally, we discuss the feasibility of $\system$ by estimating the entire process on a Raspberry Pi over a typical network and the power consumption. 
On the system model of $\system$ described in Section~\ref{sec:System}, the laptop PC corresponds to a systems manager $\SysM$, and the Raspberry Pi corresponds to an IoT device among the designated devices $\IDdsg$. 
Since a Raspberry Pi has become popular and widely used, we believe that the estimation gives us insight into $\system$ in the real world.

\subsection{Implementation and Experimental Setting} \label{sec:expereiment_setting}

We implemented the proposed CMDVS constructions in Section~\ref{sec:Construction} in the C++ language with EdDSA~\cite{eddsa} and vacuum filters~\cite{vldbWZSQ19}. 
EdDSA is implemented in the libsodium\footnote{\url{https://libsodium.gitbook.io/doc/}} library version 1.0.18-stable and the vacuum filter is implemented in the Vacuum-Filter library.\footnote{\url{https://github.com/wuwuz/Vacuum-Filter}}


We first measure the communication size when the proposed CMDVS constructions, i.e., the trivial and generic constructions, are implemented on a laptop PC. 
Our code returns a bit length per designated device via the vacuum filter library and then we count up the total size for communication with the bit length.
We also implemented the broadcast authentication~\cite{ainaWYS21} with the OpenSSL library version 1.1.1. 
The environment of the laptop PC is Ubuntu 18.04.5 LTS on the Windows Subsystem for Linux over Windows 11 and is with Intel Core i7-8565U and 16 gigabytes memory. 
The entire performance is then estimated over LoRa with its maximum transmission speed of 250 kilo-bits per second as a typical wireless network setting. 
We assume that a device identifier is 64 bits and the bit length of commands sent to designated devices is 256 bits, respectively.

\subsection{Results} \label{sec:exp_results}

\begin{figure*}[ttt!]
 \begin{minipage}[t]{0.49\hsize}
    \centering
    \hspace*{-2em}
    \includegraphics[scale=0.315]{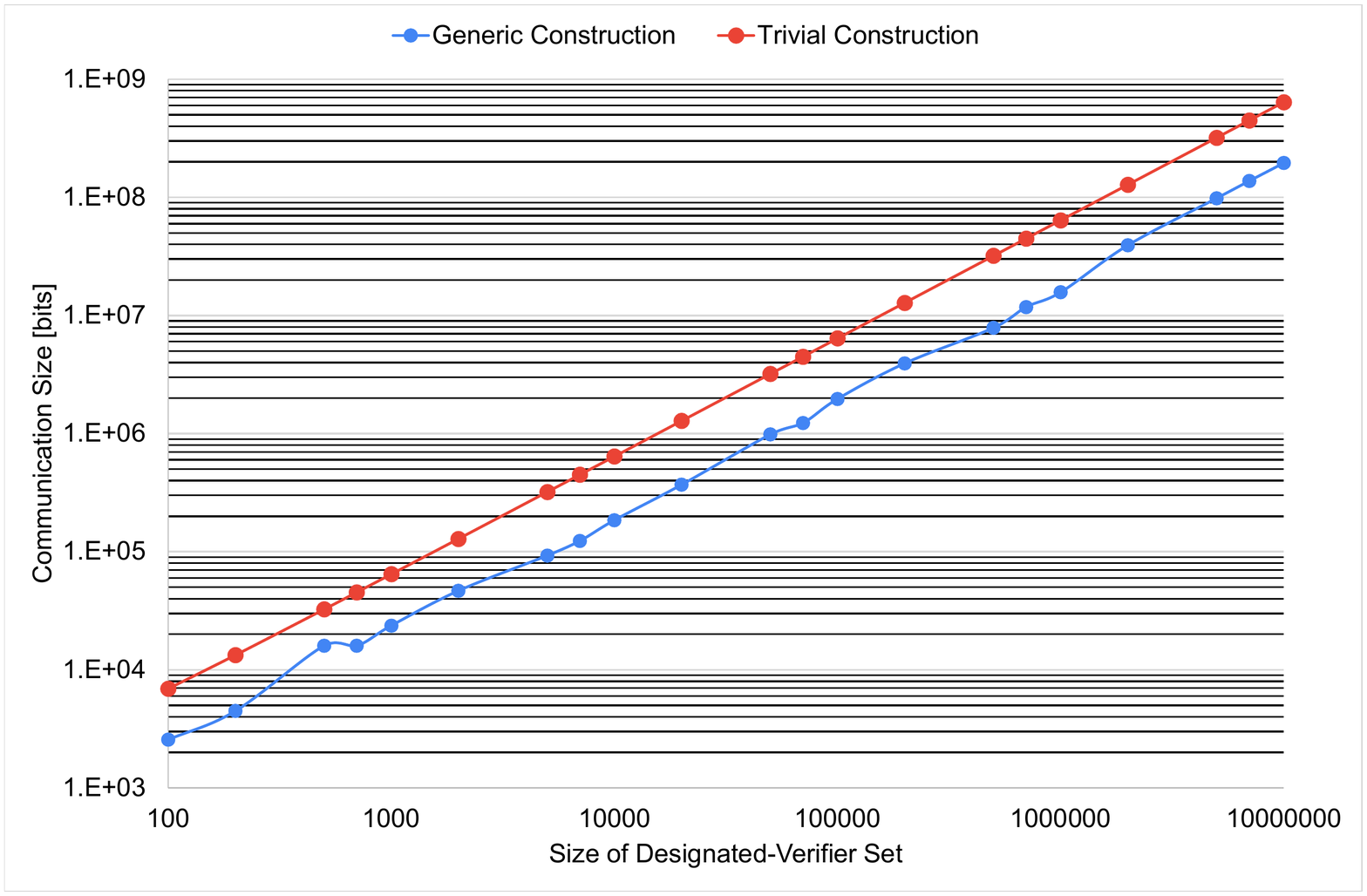} 
    \vspace*{-3.5em}
    \caption{Communication size versus the size of designated-verifier set: 
    The red line, denoted by Generic Construction, represents the proposed generic construction in Section~\ref{sec:Construction_Generic} while the blue line, denoted by Trivial Construction, represents the trivial construction in Section~\ref{sec:Construction_Trivial}, respectively. 
    }
    \label{fig:eval_datasize}
 \end{minipage}
 \hfill
 \begin{minipage}[t]{0.49\hsize}
    \centering
    \hspace*{-2em}
    \includegraphics[scale=0.315]{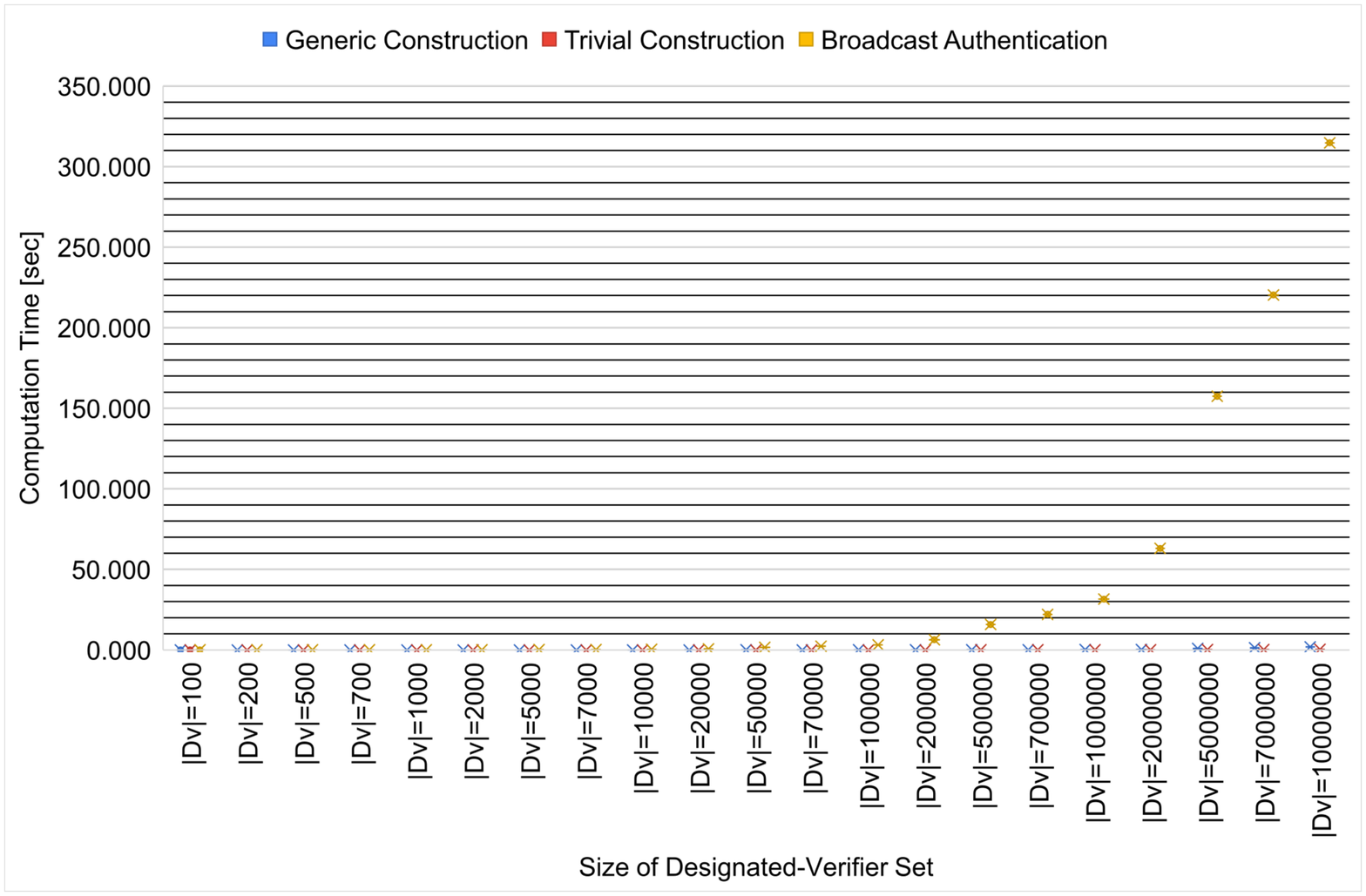}
    \vspace*{-3.5em}
        \caption{Computation time versus size of designated-verifier set for $\Sign$: This figure is a box-and-whisker plot. 
    Other setting is common with Figure~\ref{fig:eval_datasize}.
    The yellow line, denoted by Broadcast Authentication, represents the scheme in~\cite{ainaWYS21}.}
        \label{fig:eval_comp_sign_appen}
 \end{minipage}
\end{figure*}

\para{Communication Size} \label{sec:resuts_communication}
The results of the communication size are shown in Figure~\ref{fig:eval_datasize}. 
According to the figure, the communication size for the generic construction becomes four times smaller than the trivial construction and 25 times smaller than the broadcast authentication, respectively. 
Such advantage of the communication size is obtained by an AMQ structure, i.e., the vacuum filter. 
The false-positive probability of the vacuum filter is about 0.01\% in this measurement. 

The bit length per designated device for the generic construction is about 20 bits and is 
almost stable for any number of the designated devices. 
It means that the communication size could be compressed by about 30\% because the bit length per designated device for the trivial construction is 64 bits as described in Section~\ref{sec:expereiment_setting}. 
Notably, the communication size could be compressed by about 4\% compared to the broadcast authentication.

\mpara{Computation Time} \label{sec:results_computation}
We also measure the computation time for the $\Sign$ and $\Vrfy$ algorithms as shown in Figures~\ref{fig:eval_comp_sign_appen}--\ref{fig:eval_comp_vrfy}. 
For the $\Sign$ algorithm, the generic construction and the trivial construction are two orders of magnitude faster than the broadcast authentication. (See in Figure~\ref{fig:eval_comp_sign_appen}). 
Indeed, the generic construction and the trivial construction generate only a single signature, while the broadcast authentication needs to generate individual signatures in proportion to the number of devices. 
Consequently, the computation time could be drastically improved compared to broadcast authentication. 

We also compare the generic construction with the trivial construction in detail, and their results are shown in Figure~\ref{fig:eval_comp_sign} and Figure~\ref{fig:eval_comp_vrfy}, respectively. 
According to the figures, the computation times for the $\Sign$ and $\Vrfy$ algorithms of the generic construction are almost identical to those for the trivial construction until 200,000 devices. 
Meanwhile, the computation time for both $\Sign$ and $\Vrfy$ algorithms of the generic construction is greater than the trivial construction. 

The reason is that the $\Insert$ and $\Lookup$ process of the AMQ structure takes a long time in proportion to the size of a designated-verifier set $\Dv$. 
In contrast, the trivial construction needs only string operations for each algorithm, i.e., concatenation of $\Dv$ for $\Sign$ and search of $\id$ in $\Dv$ for $\Vrfy$. 
We note that the computation time for the generic construction should be longer than that for the trivial construction, because the generic construction executes the $\Insert$ and $\Lookup$ processes as well as the generation of the EdDSA signatures, whereas the trivial construction generates only the EdDSA signatures. 
The above phenomenon is common with broadcast authentication since it computes a single verification computation in the $\Vrfy$ algorithm. 

It also indicates that the overheads caused by the AMQ structure can be represented in the differences between the generic construction and the trivial construction in Figure~\ref{fig:eval_comp_sign} and Figure~\ref{fig:eval_comp_vrfy}. 
Specifically, the computation time for the $\Sign$ algorithm of the generic construction becomes about five times longer by using the AMQ structure than that of the trivial construction after 500,000 devices. 
We also note that the computation time for the $\Vrfy$ algorithm of the generic construction becomes a hundred times longer due to the use of the AMQ structure.

\begin{figure*}[ttt!]
 \begin{minipage}[t]{0.49\hsize}
    \centering
    \hspace*{-2em}
    \includegraphics[scale=0.315]{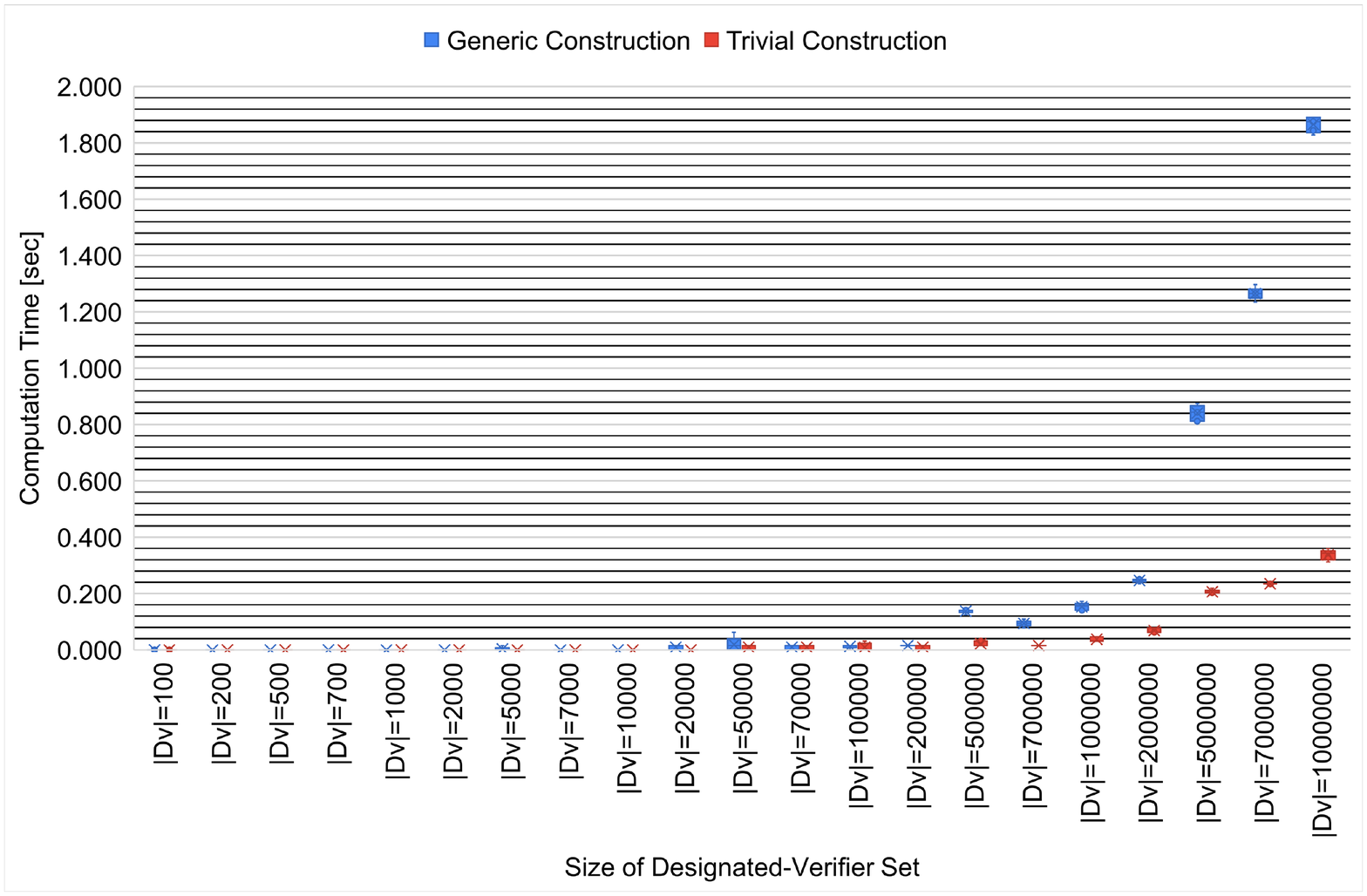}
    \vspace*{-3.5em}
    \caption{Computation time versus size of designated-verifier set for $\Sign$: This figure is the detail version of Figure~\ref{fig:eval_comp_sign_appen} excluding Broadcast Authentication.}
    \label{fig:eval_comp_sign}
 \end{minipage}
 \hfill
 \begin{minipage}[t]{0.49\hsize}
 \centering
    \hspace*{-2em}
    \includegraphics[scale=0.315]{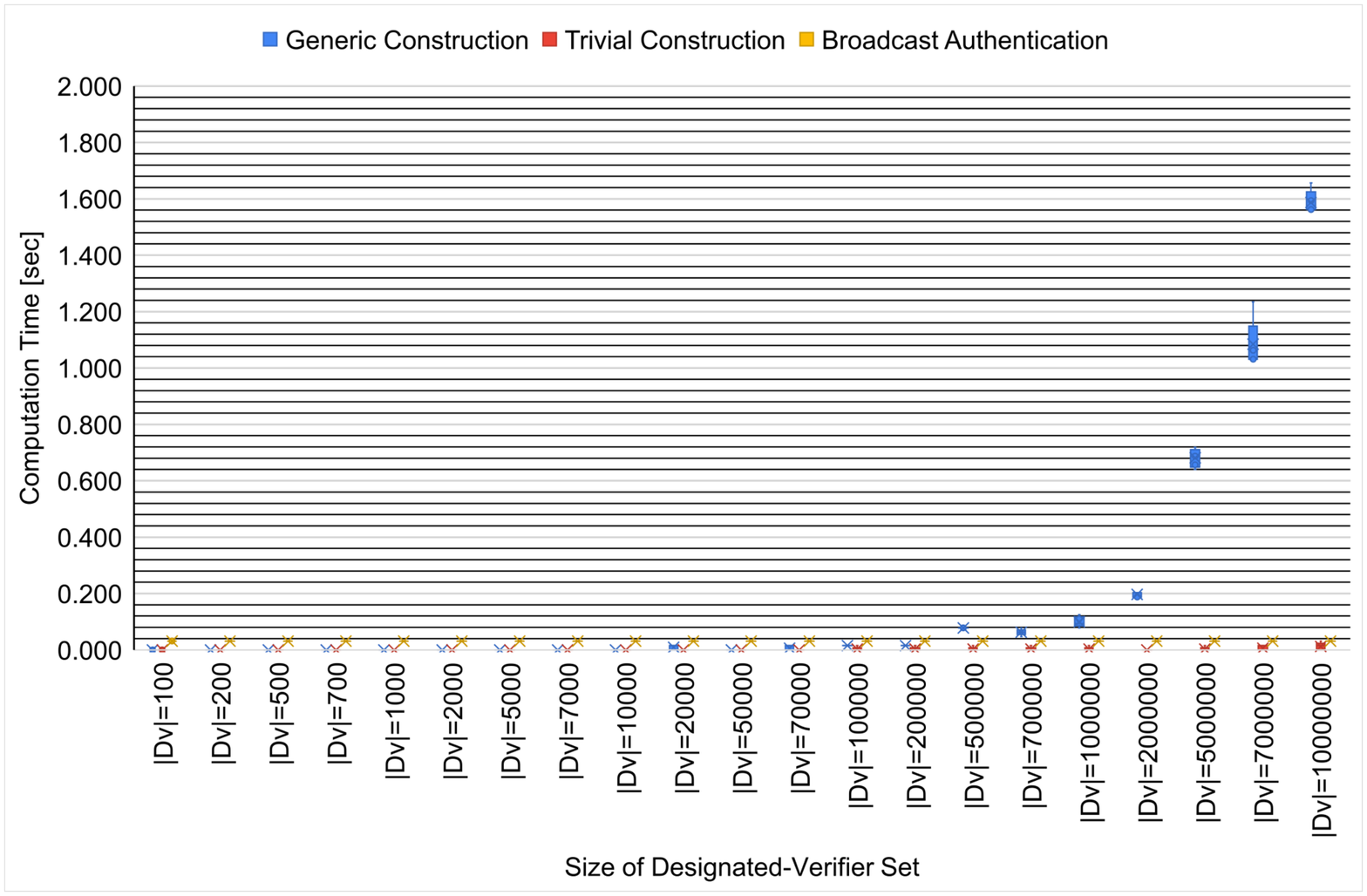}
    \vspace*{-3.5em}
    \caption{Computation time versus size of designated-verifier set for $\Vrfy$: The setting is common with Figure~\ref{fig:eval_comp_sign_appen}.}
    \label{fig:eval_comp_vrfy}
 \end{minipage}
\end{figure*}

\mpara{Entire Performance} \label{sec:entire_performance}
Based on the results in the previous subsections, the entire performance of $\system$ over the LoRa network is estimated as shown in Figure~\ref{fig:eval_final}. 
This figure shows the entire performance of $\system$ over the LoRa network, including the computation for the $\Sign$ and $\Vrfy$ algorithms, wherein a systems manager $\SysM$ generates an authenticated command $\acmd$ and each device $\id$ receives $\acmd$. 
Here, the entire performance is then estimated over LoRa with its maximum transmission speed of 250 kilo-bits per second as described above. 

According to the figure, the performance of $\system$ based on the generic construction can be three times faster than that based on the trivial construction. 
Interestingly, compared to the broadcast authentication, it is 25 times faster than the broadcast authentication, and therefore two orders of magnitude faster. 
In particular, the elapsed time per device is about 0.08 milliseconds for the generic construction, about 0.26 milliseconds for the trivial construction, and about 2 milliseconds for the broadcast authentication, respectively. 
The performance improvement is obtained by virtue of compressing the communication size via the AMQ structures. 

Since the performance improvement by the proposed construction is stable for any number of devices in $\Dv$, we can also estimate the number of IoT devices that can be controlled within a second. 
Notably, devices of more than 12,000 can be controlled by $\system$ based on the proposed construction over the LoRa network, which is greater than 4,000 devices by the trivial construction and 400 devices by the broadcast authentication.

\begin{figure*}[ttt!]
    \begin{minipage}[t]{.5\textwidth}
    \centering
    \hspace*{-2em}
    \includegraphics[scale=0.315]{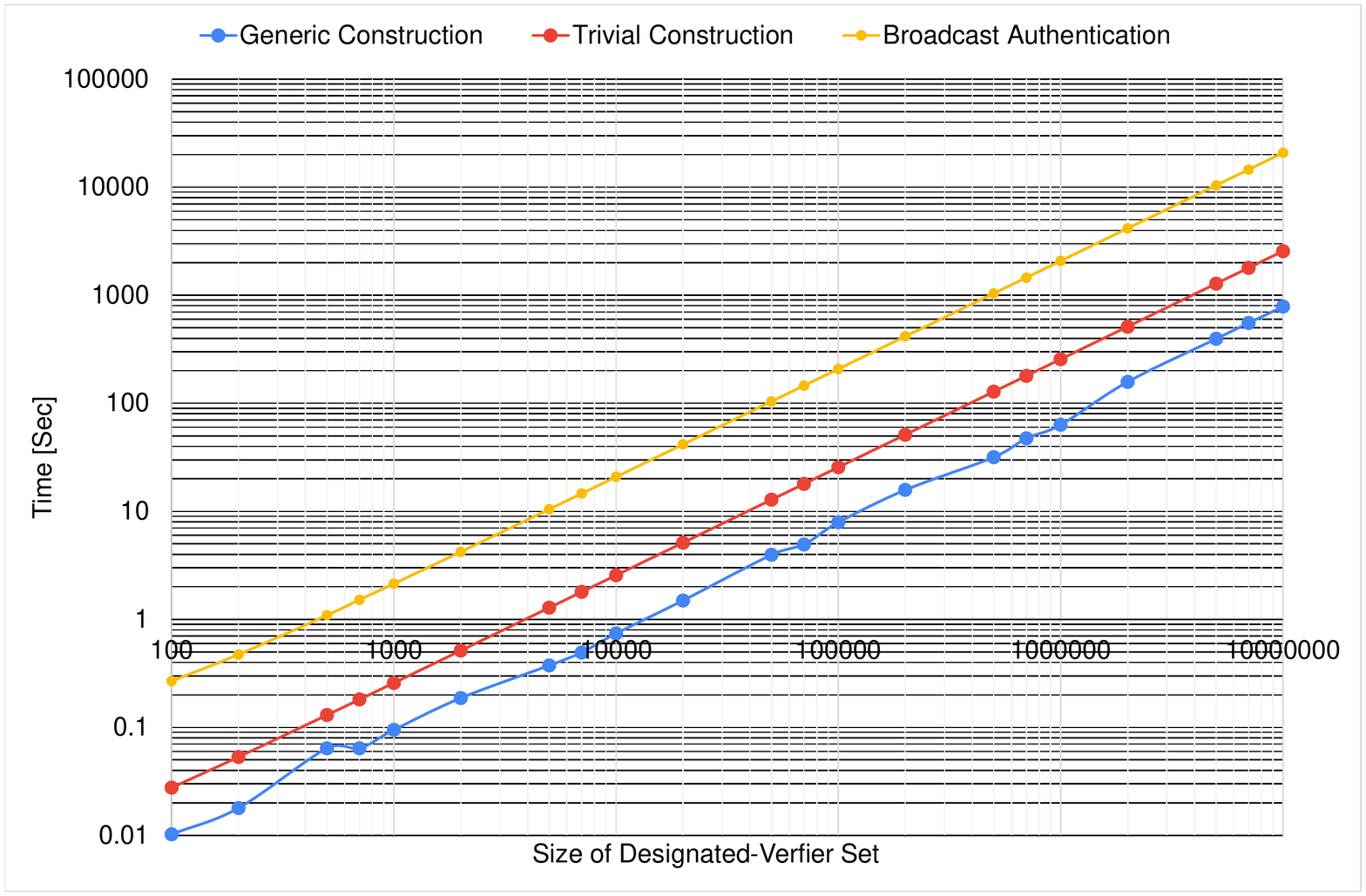}
    \vspace*{-3.5em}
    \caption{Entire performance versus the size of the designated-verifier set. 
    The setting is common with Figure~\ref{fig:eval_comp_sign}. 
    This figure includes both the communication time and the computation time. 
    }
    \label{fig:eval_final}
    \end{minipage}
    \hfill
    \def\@captype{table}

\end{figure*}

\subsection{Feasibility on IoT Devices} \label{sec:feasibility}

We discuss the feasibility of $\system$ for IoT devices in the real world. 
In particular, 
we deploy the $\Vrfy$ algorithms in a Raspberry Pi as an IoT device and then estimate the performance in 
the same setting as Section~\ref{sec:expereiment_setting}. 
We also evaluate the power consumption for battery life of IoT devices. 
The environment is with a Raspberry Pi3 with Ubuntu Server 20.4.4 LTS for the arm64 architecture.


\mpara{Entire Performance on Low-Power Devices}
We measure the computation time for the $\Vrfy$ algorithms on the Raspberry Pi and then estimate the entire performance with IoT devices as $\Dv$. 
Although we omit the detail of measurement results due to space limitation, the computation time for the generic construction and the trivial construction is almost the same until 200,000 devices, and that for the generic construction becomes greater than the trivial construction after 500,000 devices. 
In particular, the computation time for the generic construction is forty times longer than the trivial construction. 
On the other hand, it is 1.1 times longer than broadcast authentication. 
The broadcast authentication requires the Raspberry Pi to load a huge size of signatures in its memory storage. Therefore the computation time for the generic construction is close to the broadcast authentication.

Interestingly, even with the longer computation time on the Raspberry Pi, the entire performance of $\system$ over the LoRa network is almost the same as Figure~\ref{fig:eval_final}. 
The reason is that the bottleneck of $\system$ is the communication overhead as long as a low-power wide area network is utilized. 
For instance, the elapsed time for the entire process with the generic construction over the LoRa network is about 789 seconds. 
It is divided into 781 seconds for communication and 8 seconds for computation of the $\Sign$ and $\Vrfy$ algorithms. 
Similarly, the elapsed time for the entire process with the trivial construction is about 2560.4 seconds, which is divided into 2560 seconds for communication and 0.4 seconds for computation. 
The elapsed time for the entire process with the broadcast authentication is about 20800 seconds, which is divided into 20,480 seconds for communication and 320 seconds for computation.

The above fact gives us two important insights. 
First, AMQ structures are attractive because decreasing the communication size can significantly improve the entire performance, even on IoT devices. 
An IoT device can be controlled with about 0.08 milliseconds per device under $\system$ based on the generic construction. 
Second, $\system$ based on the generic construction can control devices of more than 12,000 over the LoRa network within a second. 
It is more significant than 4,000 devices by the trivial construction and 130 devices by the broadcast authentication, and we thus conclude that $\system$ based on the generic construction is practical.

\mpara{Communication Overheads on Low-Power Wide Area Networks}
We discuss $\system$ over low-power wide area networks other than LoRa as further applications. 
We know eMTC\footnote{\url{https://halberdbastion.com/technology/iot/iot-protocols/emtc-lte-cat-m1\#:\~:text=An\%20eMTC\%20Cat\%2DM1\%20network,any\%20existing\%20LTE\%20channel\%20width.}} with its maximum transmission speed of 1 mega-bits per second and SIGFOX\footnote{\url{https://www.sigfox.com/en/what-sigfox/technology}} with its maximum transmission speed of 600 bits per second as specifications for low-power wide area networks. 

$\system$ based on the generic construction is stably three times faster than the trivial construction and 25 times faster than the broadcast authentication over these networks by virtue of compressing the communication cost. 
For instance, in the case of SIGFOX, 12,000 devices are controlled within about 308 seconds by the generic construction, within about 1200 seconds by the trivial construction, and within 10240 seconds by the broadcast authentication. 
In the case of eMTC, 12,000 devices can be controlled within about 0.24 seconds by the generic construction, 0.64 seconds by the trivial construction, and 6.59 seconds by the broadcast authentication. 

Overall, for a communication protocol with its maximum transmission speed of 50 mega-bits per second, $\system$ based on the generic construction is faster than the trivial construction. 
For a communication protocol whose maximum transmission speed is 100 greater than mega-bits per second, $\system$ based on the generic construction is still faster than the trivial construction as long as the number of IoT devices is fewer than 700,000. 
Moreover, it is also faster than broadcast authentication over 5G with a maximum transmission speed of 10 gigabits per second by virtue of the use of a single signature.

\mpara{Power Consumption}
To evaluate the impact on battery lifetime for $\system$, we measured the power consumption when the codes of the CMDVS constructions are executed on the Raspberry Pi. In particular, the Raspberry Pi was connected to Watt Checker, TAP-TST10,\footnote{\url{https://www.sanwa.co.jp/en/power.html}} and then we measured the average current consumed for the codes of the $\Vrfy$ algorithms, that were executed on the Raspberry Pi. We also assume the use of Anker 633 Magnetic Battery\footnote{\url{https://www.anker.com/products/a1641?ref=search\_battery\#!}} with 10,000mAh. Here, the power consumption in the standby state of the Raspberry Pi is 1.7~watts, and a 5-volt power supply is used. The size of $\Dv$ is $10{,}000{,}000$.

The result is shown in Table~\ref{tab:power_consumption}. According to the table, the difference between the generic construction and the trivial construction is 0.1~W, while that between the generic construction and the broadcast authentication is 0.4~W, respectively. This difference seems stable even when we change the size of $\Dv$. Consequently, it is considered that the advantage is obtained by compression of the communication size through AMQ structures. When the battery described above is used, its battery life is a one-hour difference with the trivial construction and a three-hour difference with the broadcast authentication. We also note that most parts of the battery life are due to the standby state of the Raspberry Pi. When a lower-power device is used, the battery life will be longer. 

\begin{table}[t!]
    \centering
    \caption{Average on power consumption and battery life for $\system$: 
    Each value represents the average of the power consumption and battery life on five executions. 
    }
    \label{tab:power_consumption}
    
    \begin{tabular}{cccc} \hline 
          Construction & Watt & Ampere & Battery Life [h]  \\ \hline \hline
         Generic Construction & 2.4 & 0.48 & 21 \\ \hline
         Trivial Construction & 2.5 & 0.50 & 20  \\ \hline 
         Broadcast Authentication & 2.8 & 0.56 & 18 \\ \hline
    \end{tabular}
    \vspace*{-.5em}
\end{table}

\section{Related Work} 
\label{sec:RelatedWork}

\para{Cryptographic Protocols Based on AMQ Structures}
Most of the cryptographic research related to AMQ structures (e.g., \cite{ecDJSS18,cryptoNY15,ndssSSL+21}) focus on the Bloom filter~\cite{comacmBlo70} since, unlike recent experimentally-efficient AMQ structures, it has been well analyzed in a theoretical sense.
The previous works have completely different goals from ours:
Derler et al.~\cite{ecDJSS18} introduced Bloom filter encryption to efficiently realize puncturable encryption, which is a special type of public-key encryption;
Naor and Yogev~\cite{cryptoNY15} considered the Bloom filter in adversarial environments to make the Bloom filter robust;
and Sun et al.~\cite{ndssSSL+21} employed the Bloom filter to maintain an encrypted database for searchable encryption.
To the best of our knowledge, there is no research on cryptographic protocols based on AMQ structures in the context of secure remote control.

\mpara{Message Authentication Protocols for Many Users}
MDVS~\cite{icicsLV04,tccDHM+20} is digital signatures in the multi-user setting.
Each user has signing and verification keys, and any user can designate an arbitrary subset of other users and generate a signature so that only the designated users can check the validity of the signature.
The recent MDVS construction~\cite{tccDHM+20} with strong security notions require heavy cryptographic primitives such as bounded-collusion functional encryption~\cite{cryptoGVW12}.
On the other hand, CMDVS is a restricted version of MDVS and our CMDVS construction only requires AMQ structures and standard digital signatures, which are lightweight enough for IoT environments.
For the efficiency reason, we only employed our CMDVS scheme for the experimental evaluations.

Broadcast authentication~\cite{spCP10,ccsPer01,spPCTS00,tifsShi17,tcsSW01,baee2022ali} aims to broadcast a single piece of data to many receivers with data authenticity.
However, except for Watanabe et al.'s work~\cite{ainaWYS21}, the existing works do not support the functionality that a sender chooses an arbitrary subset of receivers;
data is always broadcast to all receivers.
Watanabe et al.~\cite{ainaWYS21} introduced anonymous broadcast authentication (ABA), which supports such functionality and provable anonymity.
Although ABA and CMDVS have similar functionality, they have a clear difference between them:
due to the provable anonymity, the lower bound on the authenticator sizes of ABA is $\Omega(d\cdot\secpar)$, where $d$ is the number of designated receivers and $\secpar$ is the security parameter.
Our CMDVS construction overcame the lower bound.
Note that CMDVS can be used in combination with existing (unproved) anonymizing techniques~\cite{li2015lightweightanonymous}. 

\mpara{IoT Security}
IoT security can be realized from the firmware level~\cite{costin2016automated,costin2014large-scale} to the application~\cite{fernandes2016security,ronen2017iot}. 
Although the conventional approach focuses on controlling the data flow~\cite{FlowFence,celik2019iotguard,fan2021ruledger}, cartographic approach is discussed in recent years~\cite{usenixAKAF+19,antonio2016aot,usenixKHAP+19,rana2022lightweightcryptography,seyhan2021bi-gisis}. 
To the best our knowledge, the IoT security in recent years is based on two ways~\cite{schiller2022landscape}, machine learning~\cite{nguyen2019diot,jia2017contexlot,raza2013svelte,lei2020secwir,mothukuri2022federated-learning} or trusted execution environments~\cite{xu2019dominance,suzaki2020reboot,valadares2021systematicliterature}. 
These approaches often utilize a central server to control resource-constrained IoT devices outside of them. 
In contrast, our approach is built-in for IoT devices because the $\Vrfy$ algorithm is embedded in them. 
 


\section{Concluding Remarks} \label{sec:Conclusion}
In this paper, we proposed $\system$, a secure system aiming to control IoT devices remotely.
$\system$ enables us to not only bring infected IoT devices to a halt but also have any subset of all IoT devices execute arbitrary commands.
To this end, we introduced a novel cryptographic primitive for $\system$, called centralized multi-designated verifier signatures (CMDVS).
We also provided an efficient CMDVS construction, which yields compact communication sizes and fast verification procedures for $\system$.
We further discuss the feasibility of $\system$ by implementing the CMDVS construction with vacuum filters and its experimental evaluation with a Raspberry Pi. 
We have released our source to provide reproducibility and expect further subsequent work. 
According to the evaluation results, the CMDVS construction can compress communication size to about 30\% for the trivial construction and 4\% for the broadcast authentication; hence, it is expected to $\system$ based on the CMDVS construction is three times faster than the trivial construction and 25 times faster than the broadcast authentication over typical low-power wide area networks even with an IoT device. 
Furthermore, we discussed that $\system$ is feasible with respect to the communication overheads on low-power wide area networks and the power consumption. 
We thus conclude that $\system$ based on the CMDVS construction is practical. 
We plan to conduct experiments of $\system$ in the real world for further evaluation, including physics features. 

\subsubsection*{Acknowledgments}
This research was conducted under a contract of ``Research and development on IoT malware removal / make it non-functional technologies for effective use of the radio spectrum'' among ``Research and Development for Expansion of Radio Wave Resources (JPJ000254)'', which was supported by the Ministry of Internal Affairs and Communications, Japan.


\subsubsection*{Code Availability}
Our source code is publicly available at \url{https://github.com/naotoyanai/fiilter-signature_ABA} via GitHub.

\appendix




\end{document}

%% file: preamble_usepackages.tex
\usepackage[letterpaper,hmargin=1in,vmargin=1in]{geometry}

\usepackage{authblk}
\renewcommand\Affilfont{\small}

\usepackage{amsmath, amsfonts, amssymb, mathtools,amscd}
\usepackage{amsthm}
\usepackage[noend]{algorithm,algorithmic}

\usepackage{arydshln} 
\usepackage{url}
\usepackage{ifthen}
\usepackage{bm}
\usepackage{multirow}
\usepackage{xcolor,colortbl} 
\usepackage{hhline}
\usepackage{threeparttable}
\usepackage{comment}
\usepackage{paralist,verbatim}
\usepackage{cases}
\usepackage{booktabs}
\usepackage{braket}
\usepackage{cancel} 
\usepackage{ascmac} 
\usepackage{framed}
\usepackage{graphics}
\definecolor{darkblue}{rgb}{0,0,0.6}
\definecolor{darkgreen}{rgb}{0,0.5,0}
\definecolor{maroon}{rgb}{0.5,0.1,0.1}
\definecolor{dpurple}{rgb}{0.2,0,0.65}

\usepackage[capitalise,noabbrev]{cleveref}
\usepackage[absolute]{textpos}
\usepackage{float,multicol,fancybox}
\usepackage{booktabs}
\usepackage[noend]{algorithm,algorithmic}
\usepackage{citesort}
\usepackage{subfigure}
\usepackage[margin=10pt,font=small,labelfont=bf]{caption}

%% file: preamble_newcommands.tex
\def\authnote{1}
\ifnum\authnote=1
\newcommand{\Comment}[1]{\textcolor{red}{#1}} 
\newcommand{\TODO}[1]{\medskip \noindent \textcolor{blue}{\textbf{$\langle \! \langle$ To Do: #1 $\rangle \! \rangle$}} \medskip}

\newcommand{\yohei}[1]{\textcolor{magenta}{ $\langle \! \langle$ yohei: ``#1" $\rangle \! \rangle$}}
\else
\newcommand{\Comment}[1]{} 
\newcommand{\TODO}[1]{}
\newcommand{\yohei}[1]{}
\fi

\newcommand{\para}[1]{\noindent\textbf{#1}. }
\newcommand{\mpara}[1]{\medskip\noindent\textbf{#1}. }
\newcommand{\spara}[1]{\smallskip\noindent\textbf{#1}. }

\definecolor{mygray}{gray}{0.3}

\newcommand{\ctext}[1]{\raise0.2ex\hbox{\textcircled{\scriptsize{#1}}}}

\allowdisplaybreaks 

\newcommand{\Ad}{{\sf A}}

\newcommand{\AB}{\allowbreak}
\newcommand{\ceq}{\coloneqq}

\DeclareMathAlphabet{\mathpzc}{OT1}{pzc}{m}{it}

\renewcommand{\set}[1]{\{ #1 \}} 

\newcommand{\LRceil}[1]{\left\lceil #1 \right\rceil}
\newcommand{\lrfloor}[1]{\lfloor #1 \rfloor} 
\newcommand{\LRfloor}[1]{\left\lfloor #1 \right\rfloor}

\newcommand{\LRparen}[1]{\left( #1 \right)}
\newcommand{\lrparen}[1]{( #1 )}
\newcommand{\LRbra}[1]{\left[ #1 \right]}

\newcommand{\NN}{\mathbb{N}}    
\newcommand{\RR}{\mathbb{R}}

\newcommand{\bin}{ \{ 0,1 \} }
\newcommand{\bit}[1]{\{0,1\}^{ #1 }}

\newcommand{\Hash}{\mathsf{H}}
\newcommand{\hash}{\mathsf{h}}

\newcommand{\secpar}{\kappa}
\newcommand{\msg}{\mathsf{m}}

\newcommand{\Exp}[2]{\mathsf{Exp}_{ #1}^{ #2 }}

\newcommand{\adv}{\mathsf{Adv}}
\newcommand{\Adv}[2]{\adv_{ #1}^{ #2 }}
\newcommand{\negl}[1]{\mathsf{negl}( #1 )}
\newcommand{\poly}[1]{\mathsf{poly}( #1 )}

 \newcommand{\bigO}[1]{\mathcal{O}( #1 )}

\newcommand{\system}{\textsf{IoT-REX}}
\newcommand{\Owner}{\mathsf{O}}
\newcommand{\SysM}{\mathsf{SM}}
\newcommand{\IoTvar}{\mathsf{D}}
\newcommand{\IoT}[1]{\IoTvar_{#1}}
\newcommand{\Mdlvar}{\mathsf{M}}
\newcommand{\Mdl}[1]{\Mdlvar_{#1}}

\newcommand{\cmd}{\mathsf{cmd}}
\newcommand{\acmd}{\widehat{\cmd}}

\newcommand{\ID}{\mathcal{I}}
\newcommand{\IDact}{\ID_{\mathsf{Act}}}
\newcommand{\IDdsg}{\ID_{\mathsf{Dsg}}}

\newcommand{\CMDVS}{\Pi}
\newcommand{\Setup}{{\sf Setup}}
\newcommand{\KeyGen}{{\sf KeyGen}}
\newcommand{\Sign}{{\sf Sign}}
\newcommand{\Vrfy}{{\sf Vrfy}}
\newcommand{\CMDVSall}{\CMDVS = \lrparen{\Setup, \AB \KeyGen, \AB \Sign, \AB \Vrfy}}

\newcommand{\Assign}{{\sf Assign}}

\newcommand{\id}{\mathsf{id}}
\newcommand{\M}{\mathcal{M}}				
\newcommand{\D}{\mathcal{D}}				
\newcommand{\VerS}{\mathcal{V}}				

\newcommand{\pp}{{\sf pp}}
\newcommand{\sk}{{\sf sk}}
\newcommand{\vrk}[1]{\vk_{#1}}
\newcommand{\len}{{\sf len}}

\newcommand{\flag}{\mathsf{flag}}
\newcommand{\qM}{\mathcal{Q}}

\newcommand{\UF}{\mathsf{UF}}
\newcommand{\Cons}{\mathsf{Cons}}

\newcommand{\KGOra}{\mathsf{O}_\textsc{kg}}

\newcommand{\SOra}{\mathsf{O}_\textsc{s}}

\newcommand{\KGOracle}[1]{\mathsf{O}_\textsc{kg}(\pp, \sk, #1)}

\newcommand{\SOracle}[1]{\mathsf{O}_\textsc{s}(\sk,#1)}

\newcommand{\VKList}{\mathsf{List}_\textsc{vk}}

\newcommand{\vk}{{\sf vk}}

\newcommand{\Dv}{\D_\textsc{v}}				

\newcommand{\AMQ}{\Pi_\textsc{amq}}
\newcommand{\Bloom}{\Pi_\textsc{Bloom}}
\newcommand{\Gen}{{\sf Gen}}
\newcommand{\Insert}{{\sf Insert}}
\newcommand{\Lookup}{{\sf Lookup}}

\newcommand{\AMQall}{\AMQ = \lrparen{\Gen, \AB \Insert, \AB \Lookup}}
\newcommand{\Bloomall}{\Bloom = \lrparen{\Gen, \AB \Insert, \AB \Lookup}}

\newcommand{\U}{\mathcal{U}}				
\newcommand{\X}{\mathcal{X}}				
\newcommand{\PS}{\mathcal{S}}				
\newcommand{\T}{{\sf T}}
\newcommand{\hT}{\widehat{\T}}
\newcommand{\prm}{{\sf par}}

\newcommand{\aux}{{\sf aux}}
\newcommand{\true}{\texttt{true}}
\newcommand{\false}{\texttt{false}}
\newcommand{\accept}{\texttt{accept}}

\newcommand{\fp}[1]{\mu_{ #1 }}                     

\newcommand{\Q}{\mathcal{Q}}

\newcommand{\DS}{\Pi_\textsc{ds}}
\newcommand{\SigGen}{{\sf SigGen}}
\newcommand{\SigSign}{{\sf SigSign}}
\newcommand{\SigVer}{{\sf SigVer}}

\newcommand{\DSall}{\DS = \lrparen{\SigGen, \AB \SigSign, \AB \SigVer}}

\newcommand{\sigk}{\mathsf{sigk}}
\newcommand{\verk}{\mathsf{verk}}
\newcommand{\sigmaDS}{\sigma_{\textsc{ds}}}

\newcommand{\sM}{\mathcal{M}_{\mathsf{s}}}

\newcommand{\CMA}{\mathsf{CMA}}
\newcommand{\UFCMA}{\UF\text{-}\CMA}

\newcommand{\SignOra}[1]{\mathsf{O}_{\textsc{Sig}}(\sigk, #1 )}

\newcommand{\A}{\mathsf{A}}



\makeatletter
\newcommand{\figcaption}[1]{\def\@captype{figure}\caption{#1}}
\newcommand{\tblcaption}[1]{\def\@captype{table}\caption{#1}}
\makeatother

\let\oldnl\nl
\newcommand{\nonl}{\renewcommand{\nl}{\let\nl\oldnl}}

%% file: preamble_newtheorem.tex
\newtheorem{theorem}{Theorem} 
\newtheorem{lemma}{Lemma}

\newtheorem{definition}{Definition}  
  
\newtheorem{remark}{Remark}  
\newtheorem{corollary}{Corollary}

%% file: main_body.bbl
{\small
\begin{thebibliography}{10}

\bibitem{Cisco2014}
The internet of things reference model.
\newblock Technical report, Cisco, 2014.

\bibitem{usenixAKAF+19}
M.~P. Andersen, S.~Kumar, M.~AbdelBaky, G.~Fierro, J.~Kolb, H.-S. Kim, D.~E.
  Culler, and R.~A. Popa.
\newblock {WAVE}: A decentralized authorization framework with transitive
  delegation.
\newblock In {\em {USENIX} Security'19}, pages 1375--1392. {USENIX}
  Association, 2019.

\bibitem{antonakakis2017understanding}
M.~Antonakakis, T.~April, M.~Bailey, M.~Bernhard, E.~Bursztein, J.~Cochran,
  Z.~Durumeric, J.~A. Halderman, L.~Invernizzi, M.~Kallitsis, D.~Kumar,
  C.~Lever, Z.~Ma, J.~Mason, D.~Menscher, C.~Seaman, N.~Sullivan, K.~Thomas,
  and Y.~Zhou.
\newblock Understanding the {Mirai} botnet.
\newblock In {\em {USENIX} Security'17}, pages 1093--1110. {USENIX}
  Association, 2017.

\bibitem{eddsa}
D.~J. Bernstein, N.~Duif, T.~Lange, P.~Schwabe, and B.~Yang.
\newblock High-speed high-security signatures.
\newblock {\em Journal of Cryptographic Engineering.}, 2(2):77--89, 2012.

\bibitem{bertino2017botnets}
E.~Bertino and N.~Islam.
\newblock Botnets and internet of things security.
\newblock {\em Computer}, 50(2):76--79, 2017.

\bibitem{noor2019current}
M.~binti {Mohamad Noor} and W.~H. Hassan.
\newblock Current research on internet of things (iot) security: A survey.
\newblock {\em Computer Networks}, 148:283--294, 2019.

\bibitem{comacmBlo70}
B.~H. Bloom.
\newblock Space/time trade-offs in hash coding with allowable errors.
\newblock {\em Communication of the ACM}, 13(7):422–426, jul 1970.

\bibitem{celik2019iotguard}
Z.~B. Celik, G.~Tan, and P.~D. McDaniel.
\newblock {IoTGuard}: Dynamic enforcement of security and safety policy in
  commodity iot.
\newblock In {\em NDSS 2019}, pages 1--15. Internet Society, 2019.

\bibitem{spCP10}
H.~{Chan} and A.~{Perrig}.
\newblock Round-efficient broadcast authentication protocols for fixed topology
  classes.
\newblock In {\em IEEE S\&P 2010}, pages 257--272, May 2010.

\bibitem{costin2014large-scale}
A.~Costin, J.~Zaddach, A.~Francillon, and D.~Balzarotti.
\newblock A large-scale analysis of the security of embedded firmwares.
\newblock In {\em {USENIX} Security'19}, pages 95--110. {USENIX} Association,
  2014.

\bibitem{costin2016automated}
A.~Costin, A.~Zarras, and A.~Francillon.
\newblock Automated dynamic firmware analysis at scale: A case study on
  embedded web interfaces.
\newblock In {\em ASIACCS 2016}, pages 437--448. ACM, 2016.

\bibitem{tccDHM+20}
I.~Damg{\aa}rd, H.~Haagh, R.~Mercer, A.~Nitulescu, C.~Orlandi, and S.~Yakoubov.
\newblock Stronger security and constructions of multi-designated verifier
  signatures.
\newblock In {\em TCC 2020}, pages 229--260. Springer, 2020.

\bibitem{ecDJSS18}
D.~Derler, T.~Jager, D.~Slamanig, and C.~Striecks.
\newblock {Bloom} filter encryption and applications to efficient
  forward-secret {0-RTT} key exchange.
\newblock In {\em Advances in Cryptology -- EUROCRYPT 2018}, pages 425--455.
  Springer, 2018.

\bibitem{conextFAKM14}
B.~Fan, D.~G. Andersen, M.~Kaminsky, and M.~D. Mitzenmacher.
\newblock Cuckoo filter: Practically better than {Bloom}.
\newblock In {\em CoNEXT 2014}, page 75–88. ACM, 2014.

\bibitem{fan2021ruledger}
J.~Fan, Y.~He, B.~Tang, Q.~Li, and R.~Sandhu.
\newblock Ruledger: Ensuring execution integrity in trigger-action {IoT}
  platforms.
\newblock In {\em IEEE INFOCOM 2021}, pages 1--10. IEEE, 2021.

\bibitem{fernandes2016security}
E.~Fernandes, J.~Jung, and A.~Prakash.
\newblock Security analysis of emerging smart home applications.
\newblock In {\em IEEE S\&P}, pages 636--654. IEEE, 2016.

\bibitem{FlowFence}
E.~Fernandes, J.~Paupore, A.~Rahmati, D.~Simionato, M.~Conti, and A.~Prakash.
\newblock {FlowFence}: Practical data protection for emerging {IoT} application
  frameworks.
\newblock In {\em {USENIX} Security'16}, pages 531--548. {USENIX} Association,
  2016.

\bibitem{sigmetricsGG10}
A.~Goel and P.~Gupta.
\newblock Small subset queries and {Bloom} filters using ternary associative
  memories, with applications.
\newblock In {\em ACM SIGMETRICS 2010}, page 143–154. ACM, 2010.

\bibitem{cryptoGVW12}
S.~Gorbunov, V.~Vaikuntanathan, and H.~Wee.
\newblock Functional encryption with bounded collusions via multi-party
  computation.
\newblock In {\em Advances in Cryptology -- CRYPTO 2012}, pages 162--179.
  Springer, 2012.

\bibitem{iftikhar2021resourceconstrained}
Z.~Iftikhar, Y.~Javed, S.~Y.~A. Zaidi, M.~A. Shah, Z.~Iqbal~Khan, S.~Mussadiq,
  and K.~Abbasi.
\newblock Privacy preservation in resource-constrained iot devices using
  blockchain—a survey.
\newblock {\em Electronics}, 10(14):1--26, 2021.

\bibitem{jia2017contexlot}
Y.~J. Jia, Q.~A. Chen, S.~Wang, A.~Rahmati, E.~Fernandes, Z.~M. Mao, and
  A.~Prakash.
\newblock {ContexloT}: Towards providing contextual integrity to appified {IoT}
  platforms.
\newblock In {\em NDSS 2017}, pages 1--15. Internet Society, 2017.

\bibitem{esaKM06}
A.~Kirsch and M.~Mitzenmacher.
\newblock Less hashing, same performance: Building a better {Bloom} filter.
\newblock In {\em Algorithms -- ESA 2006}, pages 456--467. Springer, 2006.

\bibitem{imaccKWS21}
H.~Kobayashi, Y.~Watanabe, and J.~Shikata.
\newblock Asymptotically tight lower bounds in anonymous broadcast encryption
  and authentication.
\newblock In {\em IMACC 2021}, pages 105--128. Springer, 2021.

\bibitem{usenixKHAP+19}
S.~Kumar, Y.~Hu, M.~P. Andersen, R.~A. Popa, and D.~E. Culler.
\newblock {JEDI}: Many-to-many end-to-end encryption and key delegation for
  {IoT}.
\newblock In {\em {USENIX} Security 19}, pages 1519--1536. {USENIX}
  Association, 2019.

\bibitem{icicsLV04}
F.~Laguillaumie and D.~Vergnaud.
\newblock Multi-designated verifiers signatures.
\newblock In {\em ICICS 2004}, pages 495--507. Springer, 2004.

\bibitem{ipsLV07}
F.~Laguillaumie and D.~Vergnaud.
\newblock Multi-designated verifiers signatures: anonymity without encryption.
\newblock {\em Inf.\ Process.\ Lett.}, 102(2):127--132, 2007.

\bibitem{lei2020secwir}
X.~Lei, G.-H. Tu, C.-Y. Li, T.~Xie, and M.~Zhang.
\newblock {SecWIR}: Securing smart home {IoT} communications via {Wi-Fi}
  routers with embedded intelligence.
\newblock In {\em MobiSys 2020}, pages 260--272. ACM, 2020.

\bibitem{li2015lightweightanonymous}
X.~Li, H.~Liu, F.~Wei, J.~Ma, and W.~Yang.
\newblock A lightweight anonymous authentication protocol using k-pseudonym set
  in wireless networks.
\newblock In {\em IEEE GLOBECOM}, pages 1--6. IEEE, 2015.

\bibitem{mothukuri2022federated-learning}
V.~Mothukuri, P.~Khare, R.~M. Parizi, S.~Pouriyeh, A.~Dehghantanha, and
  G.~Srivastava.
\newblock Federated-learning-based anomaly detection for iot security attacks.
\newblock {\em IEEE Internet of Things Journal}, 9(4):2545--2554, 2022.

\bibitem{mpitziopoulos2007defending}
A.~Mpitziopoulos, D.~Gavalas, G.~Pantziou, and C.~Konstantopoulos.
\newblock Defending wireless sensor networks from jamming attacks.
\newblock In {\em 2007 IEEE 18th International Symposium on Personal, Indoor
  and Mobile Radio Communications}, pages 1--5. IEEE, 2007.

\bibitem{cryptoNY15}
M.~Naor and E.~Yogev.
\newblock {Bloom} filters in adversarial environments.
\newblock In {\em Advances in Cryptology -- CRYPTO 2015}, pages 565--584.
  Springer Berlin Heidelberg, 2015.

\bibitem{antonio2016aot}
A.~L.~M. Neto, A.~L.~F. Souza, I.~Cunha, M.~Nogueira, I.~O. Nunes, L.~Cotta,
  N.~Gentille, A.~A.~F. Loureiro, D.~F. Aranha, H.~K. Patil, and L.~B.
  Oliveira.
\newblock {AoT}: Authentication and access control for the entire {IoT} device
  life-cycle.
\newblock In {\em ACM SenSys}, pages 1--15. ACM, 2016.

\bibitem{nguyen2019diot}
T.~D. Nguyen, S.~Marchal, M.~Miettinen, H.~Fereidooni, N.~Asokan, and A.-R.
  Sadeghi.
\newblock {D\"{I}oT}: A federated self-learning anomaly detection system for
  {IoT}.
\newblock In {\em IEEE ICDCS}, pages 756--767. IEEE, 2019.

\bibitem{sodaPPR05}
A.~Pagh, R.~Pagh, and S.~S. Rao.
\newblock An optimal {Bloom} filter replacement.
\newblock In {\em ACM-SIAM Symposium on Discrete Algorithms, SODA 2005}, page
  823–829. SIAM, 2005.

\bibitem{ccsPer01}
A.~Perrig.
\newblock The {BiBa} one-time signature and broadcast authentication protocol.
\newblock In {\em ACM CCS 2001}, pages 28--37. ACM, 2001.

\bibitem{spPCTS00}
A.~{Perrig}, R.~{Canetti}, J.~D. {Tygar}, and D.~{Song}.
\newblock Efficient authentication and signing of multicast streams over lossy
  channels.
\newblock In {\em IEEE S\&P 2000}, pages 56--73, 2000.

\bibitem{rana2022lightweightcryptography}
M.~Rana, Q.~Mamun, and R.~Islam.
\newblock Lightweight cryptography in {IoT} networks: A survey.
\newblock {\em Future Generation Computer Systems}, 129:77--89, 2022.

\bibitem{raza2013svelte}
S.~Raza, L.~Wallgren, and T.~Voigt.
\newblock Svelte: Real-time intrusion detection in the internet of things.
\newblock {\em Ad Hoc Networks}, 11(8):2661--2674, 2013.

\bibitem{baee2022ali}
M.~A. Rezazadeh~Baee, L.~Simpson, X.~Boyen, E.~Foo, and J.~Pieprzyk.
\newblock {ALI}: Anonymous lightweight inter-vehicle broadcast authentication
  with encryption.
\newblock {\em IEEE Trans.\ on Dependable and Secure Computing}, pages 1--1
  (Early Access), 2022.

\bibitem{ronen2017iot}
E.~Ronen, A.~Shamir, A.-O. Weingarten, and C.~O'Flynn.
\newblock {IoT} goes nuclear: Creating a zigbee chain reaction.
\newblock In {\em IEEE S\&P}, pages 195--212. IEEE, 2017.

\bibitem{tcsSW01}
R.~Safavi-Naini and H.~Wang.
\newblock Broadcast authentication for group communication.
\newblock {\em Theoretical Computer Science}, 269(1):1 -- 21, 2001.

\bibitem{schiller2022landscape}
E.~Schiller, A.~Aidoo, J.~Fuhrer, J.~Stahl, M.~Ziörjen, and B.~Stiller.
\newblock Landscape of {IoT} security.
\newblock {\em Computer Science Review}, 44:100467:1--18, 2022.

\bibitem{seyhan2021bi-gisis}
K.~Seyhan, T.~N. Nguyen, S.~Akleylek, K.~Cengiz, and S.~H. Islam.
\newblock {Bi-GISIS KE}: Modified key exchange protocol with reusable keys for
  {IoT} security.
\newblock {\em Journal of Information Security and Applications},
  58:102788:1--7, 2021.

\bibitem{tifsShi17}
K.~{Shim}.
\newblock Basis: A practical multi-user broadcast authentication scheme in
  wireless sensor networks.
\newblock {\em IEEE Transactions on Information Forensics and Security},
  12(7):1545--1554, July 2017.

\bibitem{ndssSSL+21}
S.~Sun, R.~Steinfeld, S.~Lai, X.~Yuan, A.~Sakzad, J.~K. Liu, S.~Nepal, and
  D.~Gu.
\newblock Practical non-interactive searchable encryption with forward and
  backward privacy.
\newblock In {\em NDSS 2021}. The Internet Society, 2021.

\bibitem{suzaki2020reboot}
K.~Suzaki, A.~Tsukamoto, A.~Green, and M.~Mannan.
\newblock Reboot-oriented iot: Life cycle management in trusted execution
  environment for disposable iot devices.
\newblock In {\em ACSAC 2020}, pages 428--441. ACM, 2020.

\bibitem{valadares2021systematicliterature}
D.~C.~G. Valadares, N.~C. Will, J.~Caminha, M.~B. Perkusich, A.~Perkusich, and
  K.~C. Gorgônio.
\newblock Systematic literature review on the use of trusted execution
  environments to protect cloud/fog-based internet of things applications.
\newblock {\em IEEE Access}, 9:80953--80969, 2021.

\bibitem{vldbWZSQ19}
M.~Wang, M.~Zhou, S.~Shi, and C.~Qian.
\newblock Vacuum filters: More space-efficient and faster replacement for
  {Bloom} and cuckoo filters.
\newblock {\em VLDB 2019}, 13(2):197–210, 2019.

\bibitem{ainaWYS21}
Y.~Watanabe, N.~Yanai, and J.~Shikata.
\newblock Anonymous broadcast authentication for securely remote-controlling
  {IoT} devices.
\newblock In {\em AINA 2021}, pages 679--690. Springer, 2021.

\bibitem{ispecWYS23}
Y.~Watanabe, N.~Yanai, and J.~Shikata.
\newblock {IoT-REX}: A secure remote-control system for iot devices from
  centralized multi-designated verifier signatures.
\newblock In {\em ISPEC 2023}. Springer, 2023.
\newblock (To appear).

\bibitem{xu2019dominance}
M.~Xu, M.~Huber, Z.~Sun, P.~England, M.~Peinado, S.~Lee, A.~Marochko,
  D.~Mattoon, R.~Spiger, and S.~Thom.
\newblock Dominance as a new trusted computing primitive for the {Internet of
  Things}.
\newblock In {\em IEEE S\&P}, pages 1415--1430. IEEE, 2019.

\bibitem{nssZAYS12}
Y.~Zhang, M.~H. Au, G.~Yang, and W.~Susilo.
\newblock (strong) multi-designated verifiers signatures secure against rogue
  key attack.
\newblock In {\em NSS 2012}, pages 334--347. Springer, 2012.

\end{thebibliography}
}
